\documentclass[journal, onecolumn]{IEEEtran}
\IEEEoverridecommandlockouts
\usepackage[T1]{fontenc}
\usepackage{cite}
\usepackage{amsmath,amssymb,amsfonts}
\usepackage{algorithmic}
\usepackage{algorithm2e}
\usepackage{mathtools}
\usepackage{graphicx}
\usepackage{textcomp}
\usepackage{enumitem}
\usepackage{bbm}
\usepackage{xcolor}
\usepackage{amsthm}
\def\BibTeX{{\rm B\kern-.05em{\sc i\kern-.025em b}\kern-.08em
    T\kern-.1667em\lower.7ex\hbox{E}\kern-.125emX}}
\graphicspath{{sections/figs}}
\RestyleAlgo{ruled}

\newcommand{\cond}{\mathchoice{\,\vert\,}{\mspace{2mu}\vert\mspace{2mu}}{\vert}{\vert}}
\DeclareMathOperator\E{\textsf{E}}
\DeclarePairedDelimiter\ceil{\lceil}{\rceil}
\DeclarePairedDelimiter{\floor}{\lfloor}{\rfloor}
\let\P\relax
\DeclareMathOperator\P{\textsf{P}}

\newcommand{\td}[1]{\tilde{#1}}
\allowdisplaybreaks
\usepackage{hyperref}
\hypersetup{
    colorlinks=true,
    linkcolor=black,
    urlcolor=black,
}

\newcommand{\hp}{\hat{p}}
\newcommand{\ha}{\hat{\alpha}}
\newcommand{\rmins}{\mathrm{ins}}
\theoremstyle{definition}
\newtheorem{definition}{Definition}
\newtheorem{remark}{Remark}

\theoremstyle{plain}
\newtheorem{theorem}{Theorem}
\newtheorem{lemma}{Lemma}
\usepackage[normalem]{ulem} 

\begin{document}

\title{Noisy Sorting Capacity\\

}

\author{Ziao Wang, Nadim Ghaddar, Banghua Zhu, and Lele Wang
\thanks{Ziao Wang is with the Department of Electrical and Computer Engineering, University of British Columbia, Vancouver, BC V6T1Z4, Canada (email: ziaow@ece.ubc.ca).}
\thanks{Nadim Ghaddar is with the Department of Electrical and Computer Engineering, University of Toronto, Toronto, ON M5S 3G8, Canada, (email: nadim.ghaddar@utoronto.ca).}
\thanks{Banghua Zhu is with the Department of Electrical Engineering and Computer Sciences, University of California Berkeley, Berkeley, CA 94720, USA, (email: banghua@cs.berkeley.edu).}
\thanks{Lele Wang is with the Department of Electrical and Computer Engineering, University of British Columbia, Vancouver, BC V6T1Z4, Canada (email: lelewang@ece.ubc.ca).}
\thanks{This work was presented in part at the 2022 IEEE International Symposium on Information Theory~\cite{isit2022paper} and in part at the 2023 IEEE International Symposium on Information Theory~\cite{isit2023paper}.}}
\maketitle

\begin{abstract}

Sorting is the task of ordering $n$ elements using pairwise comparisons. It is well known that $m=\Theta(n\log n)$ comparisons are both necessary and sufficient when the outcomes of the comparisons are observed with no noise. In this paper, we study the sorting problem when each comparison is incorrect with some fixed yet unknown probability $p$. Unlike the common approach in the literature which aims to minimize the number of pairwise comparisons $m$ to achieve a given desired error probability, we consider randomized algorithms with expected number of queries $\textsf{E}[M]$ and aim at characterizing the maximal \emph{sorting rate} $\frac{n\log n}{\E[M]}$ such that the ordering of the elements can be estimated with a vanishing error probability asymptotically. The maximal rate is referred to as the \emph{noisy sorting capacity}. In this work, we derive upper and lower bounds on the noisy sorting capacity. The two lower bounds --- one for fixed-length algorithms and one for variable-length algorithms --- are established by combining the insertion sort algorithm with the well-known Burnashev--Zigangirov algorithm for channel coding with feedback. Compared with existing methods, the proposed algorithms are universal in the sense that they do not require the knowledge of $p$, while maintaining a strictly positive sorting rate. Moreover, we derive a general upper bound on the noisy sorting capacity, along with an upper bound on the maximal rate that can be achieved by sorting algorithms that are based on insertion sort.

\end{abstract}


\section{Introduction} \label{sec:intro}

Sorting is a fundamental operation in computing. Given a set of $n$ elements, a sorting algorithm sequentially submits queries consisting of pairs of elements to be compared, and outputs a permutation that arranges the elements into some global order. In the design of a sorting algorithm, it is desirable to minimize the number of queries (or pairwise comparisons). Clearly, any sorting algorithm requires at least $\log(n!) = n\log n - n\log e + \Theta(\log n)$ queries,
since there are $n!$ permutations for a set of $n$ elements and each query reveals at most one bit of information. It is well known that using, for example, the merge sort algorithm, $n\log n + n + O(\log n)$ queries are sufficient to sort $n$ elements~\cite{Knuth1997}. Therefore, to sort a set of $n$ elements, it is both necessary and sufficient to submit $n\log n(1+o(1))$ queries.

Now suppose there is observation noise in the answers to the queries. For example, the answer to each query is flipped independently with some fixed and unknown probability $p \in (0,1/2)$. How many more queries are needed to sort $n$ elements? We aim to find the minimum number of queries such that the probability of recovering the correct permutation goes to one as $n$ tends to infinity. As we shall see in this paper, $\Theta(n \log n)$ queries is both necessary and sufficient for sorting with i.i.d. noise. Towards this end, our goal will be to characterize the \emph{sorting rate}
\begin{equation} \label{eqn:sorting_rate}
R \triangleq \frac{n \log n}{\E[M]},
\end{equation}
where $M$ denotes the number of queries, and the expectation accounts for (possibly) any randomness in the querying strategy. The sorting rate can be interpreted as the ratio of the minimum number of queries needed on average in the noiseless case to the average number of queries used in the noisy setting.


Although the sorting rate of the merge sort algorithm in the noiseless case can be made arbitrarily close to 1 when $n$ is large enough, it is not immediately obvious if, in the noisy setting, a strictly positive rate is even possible. For example, let's consider a modified version of the merge sort algorithm, where every comparison is repeated a number of times, and the outcome is taken as the majority vote of the responses. Since there are at least $k \triangleq n\log n(1+o(1))$ distinct pairwise comparisons to be made by the merge sort algorithm, there should exist at least one comparison that will be repeated at most $\left\lfloor\frac{\E[M]}{k}\right\rfloor \leq \left\lfloor\frac{1}{R}\right\rfloor \triangleq \mu$ times on average. Thus, the probability of error of this modified version of merge sort can be lower bounded by the error probability of this particular comparison, which can be further lower bounded by
\[
\sum_{i = \lceil \mu/2\rceil}^{\mu}  \binom{\mu}{i} p^i(1-p)^{\mu-i} \geq \frac{1}{\sqrt{2\mu}} \exp\left\{-\mu D\left(\frac{1}{2}  \Big\| p\right)\right\} \, \not\to 0,
\]
where $D(a \| b) \triangleq a\log\frac{a}{b} + (1-a)\log\frac{1-a}{1-b}$ and the lower bound is due to~\cite[Lemma 4.7.2]{Ash1990}. Therefore, for any non-vanishing rate $R$, a repetition-based merge sort algorithm yields a non-vanishing probability of error.

\subsection{Related Work}
The noisy sorting problem has been studied in the literature in a more general setting known as active ranking~\cite{Mohajer2017,Falahatgar2017,Shah2018,Heckel2019,Agarwal2017,lou2022active, ren2019sample}. In this setting, it is assumed that the observation noise $p_{ij}$ for the comparison of a pair of elements $i$ and $j$ is unknown and could be different for different pairs. Under this more general noise model, most of the existing work focuses on finding the optimal order for the required number of queries rather than characterizing the tight constant dependence. When specialized to the noise model considered in this work (i.e., an unknown yet fixed noise probability $p$ for all pairs of elements), the algorithms presented in~\cite{Mohajer2017,Falahatgar2017,Shah2018,Heckel2019,Agarwal2017,lou2022active} yield vanishing sorting rates. In contrast, the algorithm in~\cite{ren2019sample} can achieve a strictly positive sorting rate, an upper bound of which is shown through the black curve in Fig.~\ref{fig:comparison}.

On the other hand, the case of fixed and known $p$ (which is a special case of the setting that we consider) has been considered in~\cite{Feige1994}. The proposed sorting algorithm in~\cite{Feige1994} is based on noisy comparison trees, a key ingredient of which is the knowledge of the noise probability $p$.\footnote{In fact, the algorithm in~\cite{ren2019sample} for unknown $p$ can be viewed as a generalized version of the algorithm in~\cite{Feige1994}. In Section~\ref{sec:knownAlg}, we will derive the achievable sorting rate of~\cite{Feige1994} (which is an upper bound to the achievable rate of~\cite{ren2019sample}) and show that our proposed algorithms can achieve strictly larger rates.} In a more recent work, the authors of~\cite{gu2023optimal} fully characterized the noisy sorting capacity for the special case of known $p$.\footnote{Note that~\cite{gu2023optimal} is a follow-up work of our conference publication~\cite{isit2022paper} and appeared in preprint form after the submission of this paper.} The algorithm in~\cite{gu2023optimal} that achieves the sorting capacity is highly dependent upon the knowledge of $p$, and an extension of the algorithm for the case of unknown $p$ is not immediately obvious. In this paper, we give upper and lower bounds on the noisy sorting capacity for the more general and practical setting of unknown $p$.

Other related but different settings for noisy sorting in the literature include the cases when some distance metric for permutations is to be minimized (rather than the the probability of error)~\cite{Ailon2008, Braverman2009,Ailon2011adv,Negahban2012,Wauthier2013,Rajkumar2014,Shah2016,Shah2018,Mao2018}, when the noise for each pairwise comparison is not i.i.d. and is determined by some noise model (e.g. the Bradley--Terry--Luce model\cite{BTL1952})~\cite{Ajtai2009,Negahban2012,Rajkumar2014,Chen2015,Chen2017,ren2018}, or when the ordering itself is restricted to some subset of all permutations~\cite{Jamieson2011,Ailon-Begleiter2011}.

\begin{figure}[t]
	\centering
	\vspace{-1em}
	\hspace*{-1.25em}
	\includegraphics[scale=0.7]{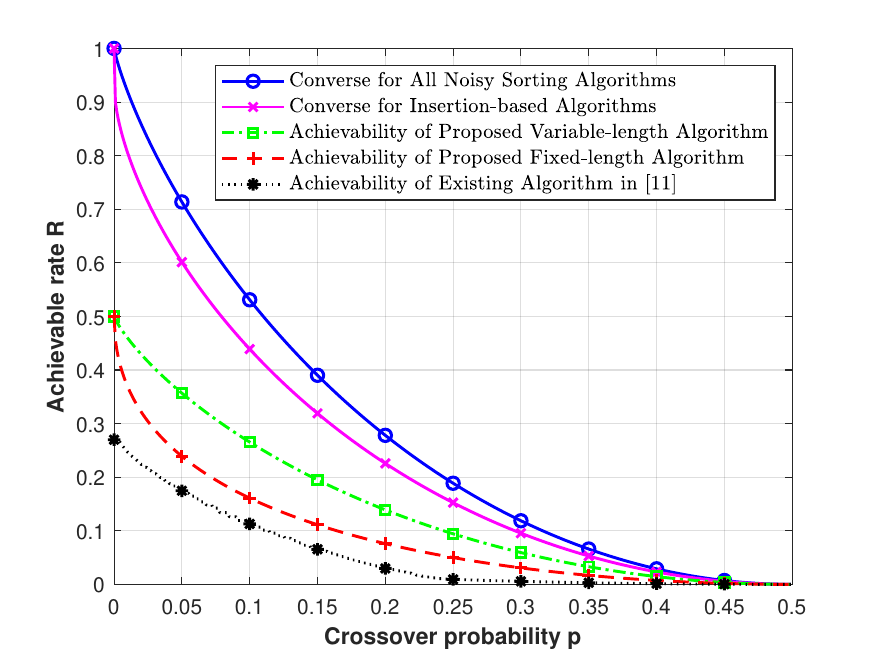}
	\caption{Comparison of achievable rates of the proposed noisy sorting algorithms and the existing algorithm, together with the converse bounds.} 
	\label{fig:comparison}
\end{figure}

\subsection{Main Contributions}
Inspired by Shannon’s formulation of the channel coding problem, the noisy sorting problem is formulated in this work as one of finding the sorting capacity, which is the maximum sorting rate $R$ such that the probability of error can be made arbitrarily small when the number of items $n$ is sufficiently large. This paper provides upper and lower bounds on the sorting capacity when the observation noise in each query is i.i.d. Bern$(p)$, and $p\in(0,1/2)$ is an unknown parameter. An intuitive upper bound on the sorting capacity follows from the feedback capacity of the $\mathrm{BSC}(p)$ channel. Since each query can convey at most $1-H(p)$ bits of information about the ordering of the elements, it takes on average $\log(n!)/(1-H(p))$ queries to convey the information contained in an unknown ordering. This leads to the upper bound $1-H(p)$ on the sorting capacity, which is shown as the solid blue curve in Fig.~\ref{fig:comparison}. We make this intuitive argument formal in Theorem~\ref{thm:vl-converse}. Conversely, the lower bounds --- one for fixed-length and one for variable-length algorithms --- are achieved by two noisy sorting algorithms that use the insertion sort algorithm along with the fixed-length and variable-length versions of the Burnashev--Zigangirov algorithm, which was introduced for the task of \emph{searching} in the presence of noise. Notably, the Burnashev--Zigangirov algorithm assumes knowledge of the parameter $p$. In the setting of unknown $p$, our proposed noisy sorting algorithms compute an estimate $\hp$ of $p$ prior to sorting, and then apply the Burnashev--Zigangirov  algorithm using the estimated parameter $\hat{p}$. A rigorous analysis of the error probability and the query complexity of the Burnashev-Zigangirov algorithm under the estimated parameter $\hat{p}$ allows to derive two achievability bounds for noisy sorting: one for the rates that can be achieved by fixed-length sorting algorithms (see the dashed red curve in Fig.~\ref{fig:comparison}), and one for the rates that can be achieved by variable-length sorting algorithms (see the dotted-and-dashed green curve in Fig.~\ref{fig:comparison}). Moreover, since the best known achievable sorting rate in the literature for the case of unknown $p$ is achieved by~\cite{ren2019sample} using an algorithm that is based upon the existing algorithm in~\cite{Feige1994}, we refine the analysis of the algorithm in~\cite{Feige1994} and show that~\cite{ren2019sample} cannot achieve a larger rate than the dotted black curve in Fig.~\ref{fig:comparison}. Hence, our proposed algorithms can achieve strictly larger rates compared to the existing literature for all values of $p$. Finally, since all noisy sorting algorithms that we considered are based upon the insertion sort algorithm, we derive an upper bound on the maximum rate that can be achieved by such a family of sorting algorithms. The upper bound is shown in the solid magenta curve of Fig.~\ref{fig:comparison}, which is strictly smaller than the $1-H(p)$ bound for general algorithms (i.e., ones that are not restricted to use the insertion sort algorithm). The bounds derived in this paper can be seen as a first step towards establishing a complete characterization of the sorting capacity when the noise probability $p$ is unknown. For a more formal description of the main results of this paper, see Section~\ref{sec:results}.

\subsection{Paper Organization and Notation}
The rest of the paper is organized as follows. In Section~\ref{sec:formulation}, we formally define the noisy sorting problem and state the definitions of noisy sorting rate and noisy sorting capacity. We present our main results in Section~\ref{sec:results}, which are achievability theorems and two two converse theorems for the noisy sorting problem. In Section~\ref{sec:achievability}, we prove the two achievability theorems by proposing two noisy sorting codes. The two converse theorems are proved in Section~\ref{sec:proof-conv}. In Section~\ref{sec:conclusion}, we conclude our results and introduce some open problems.

Notation: For an integer $n$, let $[n] = \{1,2,\ldots,n\}$. Let $\log(\cdot) = \log_2(\cdot)$. We follow the standard order notation: $f(n) = O(g(n))$ if $\lim_{n \to \infty} \frac{|f(n)|}{g(n)} < \infty$; $f(n) = \Omega(g(n))$ if $\lim_{n \to \infty}\frac{f(n)}{g(n)} >0$; $f(n) = \Theta(g(n))$ if $f(n) = O(g(n))$ and $f(n) = \Omega(g(n))$; $f(n) = o(g(n))$ if $\lim_{n \to \infty}\frac{f(n)}{g(n)} = 0$; $f(n) = \omega(g(n))$ if $\lim_{n \to \infty}\frac{|f(n)|}{|g(n)|} = \infty$; and $f(n) \sim g(n)$ if $\lim_{n \to \infty}\frac{f(n)}{g(n)} = 1$.

\section{Problem Formulation and Main Results} \label{sec:formulation-results}
\subsection{Noisy Sorting Problem} \label{sec:formulation}
Suppose an agent wishes to sort a set of $n$ elements based on noisy pairwise comparisons, where each comparison is incorrect with some fixed and unknown probability $p \in (0,1/2)$. We assume that $p$ is bounded away from $0$ and $1/2$ throughout the paper. The agent submits a query consisting of a pair of indices $i$ and $j$, and receives a noisy response of whether element $i$ is ranked higher or lower than element $j$. The goal is to estimate the total ordering of the elements\footnote{For simplicity, we assume that the ranking metrics among the elements are distinct, i.e., the total ordering is unique.}. More formally, let $\{\theta_1,\ldots,\theta_n\}$ be a finite set of real numbers to be sorted. An ordering is characterized by a permutation mapping $\pi:[n]\to[n]$ such that
\[
\theta_{\pi(1)} < \theta_{\pi(2)} < \cdots < \theta_{\pi(n)}.
\]
At the $k$th time step, the agent submits the query $(U_k,V_k) \triangleq (\theta_i,\theta_j)$ for some $i,j \in [n]$, $i \neq j$, and gets the response $Y_k$ which can be expressed as
\[
Y_k = \mathbbm{1}_{\{U_k < V_k\}} \oplus Z_k,
\]
where $\mathbbm{1}_{\{\cdot\}}$ denotes the indicator function, and $Z_k \sim \mathrm{Bern}(p)$ is a sequence of i.i.d. random variables. The agent aims to compute an estimate $\hat{\pi}$ of $\pi$ based on the responses.

\begin{definition}
An \emph{adaptive} querying strategy $\{f_k\}_{k\geq 1}$ associated with a stopping time (i.e., number of queries) $M$ is a causal protocol for determining the queries
\[
    (U_k,V_k) = f_k(Y^{k-1}, U^{k-1}, V^{k-1})
\]
in the $k$th time step. If the number of queries $M$ is predetermined (i.e., independent of the query responses), we say that the querying strategy is \emph{fixed-length}. Otherwise, the strategy may have a variable stopping time which can depend on the query responses. 
\end{definition}

\begin{definition}
   An $(R,n)$ \emph{noisy sorting code} for the noisy sorting problem consists of an adaptive querying strategy $\{f_k\}_{k\geq 1}$ associated with a stopping time $M$, and an estimator $\hat{\Pi} = \hat{\pi}(Y^M,U^M,V^M)$, where $\E[M] = \frac{n\log n}{R}$ is the expected number of queries, and $R$ is the sorting rate. When the querying strategy is fixed-length, we refer to the code as an $(R,n)$ \emph{fixed-length noisy sorting code}.
\end{definition}



\begin{definition}
A sorting rate $R$ is said to be \emph{achievable} for the noisy sorting problem if there exists a sequence of $(R,n)$ noisy sorting codes such that
\[
	\lim\limits_{n\to \infty}\max_{\pi}\P\{\hat{\Pi} \neq \pi\} = 0.
\]
\end{definition}

\begin{definition}
A sorting rate $R_\mathrm{f}$ is said to be \emph{fixed-length-achievable} for the noisy sorting problem if there exists a sequence of $(R_\mathrm{f},n)$ fixed-length noisy sorting codes such that
\[
	\lim\limits_{n\to \infty}\max_{\pi}\P\{\hat{\Pi} \neq \pi\} = 0.
\]
\end{definition}

\begin{definition}
The \emph{noisy sorting capacity} $C(p)$ is defined as the supremum over all achievable sorting rates for the noisy sorting problem when the observation noise is i.i.d. $\mathrm{Bern}(p)$. The \emph{fixed-length noisy sorting capacity} $C_\mathrm{f}(p)$ is defined analogously.
\end{definition}

\begin{remark}
Because fixed-length noisy sorting codes are special instances of noisy sorting codes, it follows that $C_\mathrm{f}(p) \leq C(p)$ for any $p\in(0,1/2)$.
\end{remark}


\begin{remark}[Reduction to average error probability]\label{rem:reduction}
    Suppose we have an $(R,n)$ noisy sorting code that achieves an \emph{average} error probability $\P(\hat{\Pi}\neq\Pi)\triangleq \epsilon$ for a uniformly chosen permutation $\Pi \sim \mathrm{Unif}(\mathcal{S}_n)$. Such a code can be used to construct an $(R,n)$ noisy sorting code achieving a \emph{worst case} error probability identical to the error of the average code, i.e., $\max_{\pi}\P(\hat{\Pi}\neq \pi) = \epsilon$. To see this, consider a uniformly chosen permutation $\bar{\Pi}$ and use it to rename the items $\theta_1,\ldots,\theta_n$ by $\theta_{\bar{\Pi}(1)},\ldots,\theta_{\bar{\Pi}(n)}$. Then, the noisy sorting code achieving an average error probability can be used to estimate the composed permutation $\bar{\Pi}\circ \pi$, which has a uniform distribution. Thus, we have that $\P(\hat{\Pi}\neq (\bar{\Pi}\circ \pi) )= \epsilon$. Since we know $\bar{\Pi}$, if we correctly estimate $\hat{\Pi}=\bar{\Pi}\circ \pi$, we can always retrieve $\pi$ by applying the inverse of $\bar{\Pi}$ to $\hat{\Pi}$. Therefore, when we analyze the worst case error probability of a noisy sorting code, we can assume without loss of generality that the true permutation $\pi$ has a uniform prior.

\end{remark}

\begin{algorithm}[t]
\caption{Insertion-Based Noisy Sorting Codes} \label{alg:ins}
\DontPrintSemicolon
\SetKwInOut{Input}{Input}
\SetKwInOut{Output}{Output}
\Input{Unordered elements $\theta_1,\ldots,\theta_n$.}
\Output{Permutation estimate $\hat{\pi}$.}
\nl Initialize $\sigma$ as a $1\times 1$ array with element $n$.\\
\For{$i=1,\ldots, n-1$}{
\nl    $\hat{l} \gets \texttt{SEARCH}\big((\theta_{\sigma(1)},\ldots,\theta_{\sigma(i)}), \: \theta_{n-i}\big)$.\\
\nl    $\sigma \gets \big( \sigma(1:\hat{l}-1), \:\, n-i, \:\, \sigma(\hat{l}:\text{end}) \big)$.
}
\nl $\hat{\pi}\gets \sigma$. \tcp*{Treat the $1\times n$ array $\sigma$ as a permutation}
\nl \Return{$\hat{\pi}$}
\end{algorithm}

We are also interested in a family of \emph{insertion-based noisy sorting codes}, which we define next. These codes are inspired by the insertion sort algorithm for sorting in the noiseless case. This algorithm successively inserts each element in its estimated position with respect to the previously inserted elements. That is, the algorithm first inserts $\theta_{n-1}$ in its estimated position with respect to $\theta_{n}$. Then, $\theta_{n-2}$ is inserted in its estimated position with respect to $\{\theta_{n-1},\theta_n \}$, and so on, until all elements are inserted. Hence, in the $i$th iteration step, the algorithm runs an instance of a searching algorithm to insert $\theta_{n-i}$ in its position with respect to the previously ordered list. In the presence of noise, a similar technique can be followed using a suitable \emph{noisy searching} algorithm. A pseudocode describing the family of insertion-based noisy sorting codes is given in Algorithm~\ref{alg:ins}, where $\texttt{SEARCH}\big((\theta_{\sigma(1)},\ldots,\theta_{\sigma(i)}), \: \theta_{n-i}  \big)$ corresponds to the noisy searching algorithm which returns the position $\hat{l}$ of $\theta_{n-i}$ within the ordered list $(\theta_{\sigma(1)},\ldots,\theta_{\sigma(i)})$, assuming the ordering generated from the previous insertions is correct. Note that different algorithms in the family of insertion-based noisy sorting codes correspond to different noisy searching algorithms. The sorting rate of an insertion-based noisy sorting code is defined as in~(\ref{eqn:sorting_rate}), where $\E[M]$ is the expected number of queries across all searching iterations.


\begin{definition}
A sorting rate $R_{\mathrm{ins}}$ is said to be \emph{insertion-sort-achievable} for the noisy sorting problem if there exists a sequence of $(R_{\mathrm{ins}}, n)$ insertion-based noisy sorting codes such that
\[
	\lim\limits_{n\to \infty}\max_{\pi} \P\{\hat{\Pi} \neq \pi\} = 0.
\]
\end{definition}

\begin{definition}
The \emph{insertion-based noisy sorting capacity} $C_{\mathrm{ins}}(p)$ is defined as the supremum over all insertion-sort-achievable sorting rates for the noisy sorting problem when the observation noise is i.i.d. $\text{Bern}(p)$.
\end{definition}

\begin{remark}
Since insertion-based noisy sorting codes are special instances of noisy sorting codes, it follows that $C_{\mathrm{ins}}(p) \leq C(p)$ for any $p\in(0,1/2)$.
\end{remark}

In this work, we derive upper and lower bounds on the sorting capacity $C(p)$. We also derive a lower bound on the fixed-length noisy sorting capacity $C_\mathrm{f}(p)$ and an upper bound on the insertion-based noisy sorting capacity $C_{\mathrm{ins}}(p)$. Note that under this information theoretic formulation of the noisy sorting problem, both the number of elements and the expected number of queries are taken to infinity, while keeping the ratio of the two (i.e., the sorting rate) fixed. The goal here is to characterize the maximum sorting rate that can allow to estimate the permutation $\pi$ with a vanishing error probability asymptotically. Note that this formulation is quite different from the approach taken in the existing literature, where the number of elements is considered to be fixed, and the goal is to minimize the expected number of queries to achieve some desired error probability. One can see that our formulation of the problem is similar in nature to Shannon's formulation of the channel coding problem.

\subsection{Main Results} \label{sec:results}
The main results of this paper can be summarized in four theorems. The first two theorems are achievability results. They provide lower bounds for the noisy sorting capacity and fixed-length noisy sorting capacity respectively. Both bounds strictly improve the rate attained by the algorithm in~\cite{Feige1994}. 
The other two theorems are converse results. They provide upper bounds for noisy sorting capacity and insertion-based noisy sorting capacity respectively.

\begin{theorem}[Achievability] \label{thm:vl-achievability}
Any sorting rate
$$R < \frac12(1-H(p))$$
 is achievable for the noisy sorting problem using the noisy sorting algorithm presented in Section~\ref{sec:proof-vl-achievability}, where $H(\cdot)$ denotes the binary entropy function. In other words, $C(p) \geq \frac12(1-H(p))$.
\end{theorem}

\begin{theorem}[Achievability of fixed-length noisy sorting codes] \label{thm:achievability}
Any sorting rate
\[
R_\mathrm{f} < \frac{1}{2}\log(1/g(p)),
\]
where $g(p) \triangleq \frac{1}{2}+\sqrt{p(1-p)}$, is achievable for the noisy sorting problem using the fixed-length noisy sorting algorithm presented in Section~\ref{sec:proof-achievability}. In other words, $C_\mathrm{f}(p) \geq \frac{1}{2}\log(1/g(p))$.
\end{theorem}
Theorems~\ref{thm:vl-achievability} and \ref{thm:achievability} will be proved in Section~\ref{sec:achievability}, by giving specific constructions of noisy sorting codes. Both constructions that will be discussed are based on algorithms designed for the \emph{noisy searching problem}, the problem of locating an element within a sorted sequence when noise is present. Before we describe the code constructions in Section~\ref{sec:achievability}, we briefly introduce the noisy searching problem in Section~\ref{sec:noisy-search}, while setting up the necessary notation. 
\begin{theorem}[Converse] \label{thm:vl-converse}
The noisy sorting capacity $C(p)$ is upper bounded as
$$C(p)\leq 1-H(p).$$
\end{theorem}
\begin{theorem}[Converse for insertion-based noisy sorting codes]\label{thm:insert-sort-conv}
The insertion-based noisy sorting capacity $C_\mathrm{ins}(p)$ is upper bounded as
$$C_{\mathrm{ins}}(p)\le \frac{1}{\frac{1}{1-H(p)}+\frac{1}{D(p\|1-p)}}.$$
\end{theorem}

Intuitively, the result of Theorem~\ref{thm:vl-converse} is expected, since every query can give at most $1-H(p)$ bits of information about the unknown permutation $\pi$. Thus, any sequence of $(R,n)$ noisy sorting codes with vanishing error probability should have
\[
\log (n!) \le \E[M](1-H(p)),
\]
which directly implies the statement of Theorem~\ref{thm:vl-converse}. The proof of both converse theorems are based on a converse for the noisy searching problem described in section~\ref{sec:noisy-search}. We defer the formal proofs of Theorems~\ref{thm:vl-converse} and~\ref{thm:insert-sort-conv} to Section~\ref{sec:proof-conv}.

\subsection{Detour: Noisy Searching Problem} \label{sec:noisy-search}

In the noisy searching problem, an agent aims to identify the position of a given real number $\tilde{\theta}$ within a sorted sequence $(\tilde{\theta}_0,\tilde{\theta}_1,\ldots,\tilde{\theta}_n)$, where $\tilde{\theta}_0 = -\infty$ and $\tilde{\theta}_n = \infty$\footnote{Throughout the paper, we use the ``tilde'' mark over variables that are associated with the noisy searching problem.}. That is, the goal of the agent is to identify the index $l^{*}$ such that
\[
\tilde{\theta}_{l^{*}-1} < \tilde{\theta} < \tilde{\theta}_{l^{*}}.
\]
To aid in the exposition, denote by $I_i$ the $i$th sub-interval among the sorted sequence, i.e., $I_i = (\tilde{\theta}_{i-1},\tilde{\theta}_i]$ for each $i\in [n-1]$, and $I_n = (\tilde{\theta}_{n-1},\tilde{\theta}_n)$.
At the $k$th time step, the agent gives the system a query $\td{U}_k\triangleq \td{\theta}_i$ for some $i\in [n-1]$, and receives a noisy response $\td{Y}_k$ of whether $\td{U}_k$ lies to the left or right of $\tilde{\theta}$. The response $\td{Y}_k$ can be expressed as
\begin{equation} \label{eqn:noisy-search}
\td{Y}_k = \mathbbm{1}_{\{\td{U}_k < \tilde{\theta}\}} \oplus \td{Z}_k,
\end{equation}
where $\td{Z}_k \sim \mathrm{Bern}(p)$ is a sequence of i.i.d. random variables.

Each variable-length noisy searching algorithm consists of an adaptive \emph{querying strategy} $\{\td{f}_k\}_{k=1}^\infty$, which is a causal protocol for determining the queries
\[
    \td{U}_k = \td{f}_k(\td{Y}^{k-1}, \td{U}^{k-1}).
\]
The adaptive querying strategy is associated with a stopping time $\td{M}$. Its output is an estimate $\hat{L}=\hat{l}(\td{Y}^{\td{M}},\td{U}^{\td{M}})$ for the location $l^*$. Notice that the random stopping time $\td{M}$ represents the number of the queries made by the algorithm. A desirable goal of designing noisy searching algorithm is to minimize the expected number of queries $\E[\td{M}]$ while keeping the error probability $\P(\hat{L}\neq l^*)$ small.

The noisy searching problem was first introduced by R\'enyi~\cite{Renyi1961} and Ulam~\cite{Ulam1976} in the context of a ``two-player 20-question game with lies'', and further developed by Berlekamp~\cite{Berlekamp1964} and Horstein~\cite{Horstein1963} in the context of coding for channels with noiseless feedback. The resemblance to the coding problem holds by viewing the position of $\theta$ as a message that needs to be communicated over a binary symmetric channel with crossover probability $p$, and the number of queries submitted by the agent as the number of channel uses. The availability of noiseless feedback pertains to the fact that the agent can use the previous received responses to adapt the querying strategy. If the agent is allowed to query any point on the real line, then Horstein's coding scheme~\cite{Horstein1963} suggests a querying strategy for the noisy searching problem: roughly speaking, query the point that equally bisects the posterior of the target point $\theta$ given the responses received so far. In fact, Shayevitz and Feder have shown in~\cite{Shayevitz2011} the optimality of Horstein's coding scheme for the channel coding problem. 
We discuss the detailed connection between the noisy searching problem and the channel coding with feedback problem in Section~\ref{sec:proof-conv}.


The noisy searching problem has been also studied in settings where the agent is restricted to query one of the elements $\theta_i$ of the sorted sequence. It is known in the literature that an expected $\frac{\log n}{1-H(p)}(1+o(1))$ queries are both necessary and sufficient to achieve vanishing error probability~\cite{Burnashev1974,Burn1976,Ben-Or2008,Javidi2015}. In~\cite{Ben-Or2008,Javidi2015,Burnashev1974}, variable-length algorithms achieving vanishing error probability with the above optimal expected number of queries are proposed, while~\cite{Burnashev1974} also provides fixed-length construction for the algorithm. In the following, we will describe the both the variable-length and fixed-length noisy searching algorithms of~\cite{Burnashev1974}, show how they can be used to construct noisy sorting codes for the sorting problem, and derive their achievable rates.

\section{Proof of Achievability results}\label{sec:achievability}
In this section, we describe our achievability results for noisy sorting. First, we prove Theorem~\ref{thm:achievability} by proposing a fixed-length noisy sorting code and analyze its achievable rate. Then, we prove Theorem~\ref{thm:vl-achievability} by modifying the fixed-length code into a variable-length version and analyze the achievable rate.
Moreover, we obtain an upper bound for the rate of the best-known algorithm in the unknown $p$ setting proposed in~\cite{ren2019sample} through a refined analysis. We show that the achievable rates by the proposed noisy sorting codes are strictly higher than that of the existing noisy sorting code in~\cite{ren2019sample}.
\subsection{Achievability: Fixed-length Noisy Sorting Code}
\label{sec:proof-achievability}
In this section, we prove Theorem~\ref{thm:achievability} by describing a fixed-length noisy sorting code that can achieve any sorting rate $R < \frac{1}{2}\log(1/g(p))$. A key ingredient of the proposed sorting code is the Burnashev--Zigangirov (BZ) algorithm for the noisy searching problem~\cite{Burnashev1974}, which we delineate below.

\vspace{.5em}
\emph{\textbf{Burnashev--Zigangirov Algorithm for Noisy Searching:}}
Recall the noisy searching problem outlined in Section~\ref{sec:noisy-search}. Unlike Horstein's scheme which queries the median of the posterior distribution of the target point $\theta$, the BZ algorithm queries one of the two endpoints of the sub-interval containing the median. With this slight modification, Horstein's scheme can be generalized for searching over a finite set (i.e., the set of sub-intervals) with a tractable analysis of the error probability.


To this end, for each $k\in[0,\ldots,\tilde{m}]$, we define a pseudo-posterior distribution $q_k(\cdot)$, where $q_k(i)$ denotes a confidence score of $\tilde{\theta}$ belonging to interval $I_i$ for each $i$ given the previous $k$ queries $\td{Y}_1,\ldots,\td{Y}_{k}$ and responses $\td{U}_1,\ldots,\td{U}_{k}$. As presented in later parts, $q_k(\cdot)$ is a probability distribution on the $n$ possible intervals for each fixed $k$, and the update of $q_k(\cdot)$ from $q_{k-1}(\cdot)$ follows the Bayes' Rule except that the crossover probability $p$ for each query is replaced by an input parameter $\alpha$ to the algorithm. The prior $q_0(\cdot)$ is initialized to the uniform distribution. There are two reasons for defining the pseudo-posteriors instead of setting $q_k(\cdot)$ as the true posterior probability $\P\{\td{\theta} \in I_i \cond \td{Y}_1,\ldots,\td{Y}_{k},\td{U}_1,\ldots,\td{U}_{k}\}$. Firstly, $p$ is unknown in our setting, so we cannot update the posteriors exactly following the Bayes' Rule. Secondly, the analysis of the fixed-length BZ algorithm~\cite{castro2008active} necessitates this update with parameter $\alpha>p$, while no performance guarantee is given when $\alpha=p$.

Let $j(k)$ be the index of the sub-interval containing the median of the $(k-1)$th pseudo-posterior distribution, i.e., $j(k)$ is the index s.t.
\begin{equation} \label{eqn:j(k)}
    \sum_{i=1}^{j(k)-1}q_{k-1}(i) \leq \frac{1}{2} \quad \text{and} \quad \sum_{i=1}^{j(k)}q_{k-1}(i) > \frac{1}{2}.
\end{equation}
The querying strategy of the BZ algorithm is to choose $\td{U}_{k}$ randomly among $\{\td{\theta}_{j(k)-1},\td{\theta}_{j(k)}\}$, as follows. Let
\begin{equation}\label{eqn:jstar}
j^{*}(k) = \begin{cases}
j(k)-1 &\:\:\: \text{with probability } p^{*}(k),\\
j(k) &\:\:\: \text{with probability } 1-p^{*}(k),
\end{cases}
\end{equation}
where
\[
p^{*}(k) = \frac{2\sum_{i=1}^{j(k)}q_{k-1}(i) - 1}{2q_{k-1}(j(k))},
\]
and set $\td{U}_{k} = \td{\theta}_{j^{*}(k)}$. The remaining part of the algorithm is to update the pseudo-posteriors $q_k(i)$ after receiving the noisy response $\td{Y}_{k} = \td{y}_{k}$
as
\begin{equation} \label{eqn:posterior}
    \begin{aligned}
    q_{k}(i) = \begin{cases}
    \frac{2(1-\alpha) q_{k-1}(i)}{1+(1-2\td{y}_{k})(1-2\alpha)\tau(k)} &\:\:\: \text{if } \td{y}_{k} = \mathbbm{1}_{\{j^{*}(k)\leq i\}}, \\[5pt]
    \frac{2\alpha q_{k-1}(i)}{1+(1-2\td{y}_{k})(1-2\alpha)\tau(k)} &\:\:\: \text{otherwise},
    \end{cases}
    \end{aligned}
\end{equation}
where
\[
\tau(k) = 2\sum_{i=1}^{j^{*}(k)}q_{k-1}(i) - 1.
\]
It is not hard to check that when $\alpha$ is set to $p$, this updates exactly follow the Bayes' Rule, and the pseudo-posterior $q_{k}(i)$ is exactly the posterior probability $\P\{\td{\theta} \in I_i \cond \td{Y}_1,\ldots,\td{Y}_{k},\td{U}_1,\ldots,\td{U}_{k}\}$.
After $\td{m}$ queries, the output of the BZ algorithm is the index $i$ of the sub-interval with the largest posterior $q_{\td{m}}(i)$.

A pseudocode of the BZ algorithm is given in Algorithm~\ref{alg:bz}. The following lemma bounds the error probability of BZ algorithm.

\begin{lemma}[Theorem 8.1 in~\cite{castro2008active}]
\label{thm:fl-bz}
    Suppose $p<\alpha<1/2$, then the error probability of the BZ algorithm satisfies
    \begin{equation}
    \label{eqn:bz_error}
         \P\{\td{\theta} \notin I_{\hat{l}}\} \leq (n-1)\left(\frac{1-p}{2(1-\alpha)}+\frac{p}{2\alpha}\right)^{\td{m}}.
    \end{equation}
\end{lemma}


\begin{algorithm}[t]
\caption{Fixed-length Burnashev--Zigangirov \texttt{BZ} Algorithm } \label{alg:bz}
\SetKwInOut{Input}{Input}
\SetKwInOut{Output}{Output}
\Input{Ordered sequence $(\td{\theta}_1,\ldots,\td{\theta}_{n-1})$, target $\td{\theta}$, number of queries $\td{m}$, parameter $\alpha$. }
\Output{Index $\hat{l}$ of the sub-interval containing $\td{\theta}$.}
\nl $q_0(i) \gets 1/n$ for each $i \in [n]$.\\
\For{$k=1,\ldots,\td{m}$} {
\nl     Find $j(k)$ using inequalities in (\ref{eqn:j(k)}).\\
\nl     Set $j^{*}(k)$ according to~\eqref{eqn:jstar}.\\
\nl     Set $\td{U}_k \gets \td{\theta}_{j^{*}(k)}$.\\
\nl     Receive noisy response $\td{Y}_{k}$.\\
\nl     Update $q_k(i)$ for each $i\in [n]$ as in (\ref{eqn:posterior}).
}
\nl $\hat{l} \gets \arg\max_i \, q_{\td{m}}(i)$.\\
\nl \Return{$\hat{l}$}
\end{algorithm}

\emph{\textbf{Proposed Fixed-length Noisy Sorting Code:}}
The proposed fixed-length noisy sorting code uses the well-known insertion sort along with the BZ algorithm. More specifically, for each $i \in \{1,\ldots,n-1\}$, the fixed-length BZ algorithm is used to \emph{insert} $\theta_{n-i}$ in the correct location within the already-sorted sequence of elements. That is,
$\theta_{n-1}$ is initially inserted in its estimated position with respect to $\theta_{n}$, 
and then $\theta_{n-2}$ is inserted in its estimated position with respect to now-known ordering of $\theta_{n-1}$ and $\theta_{n-2}$,
and so on. In the $i$th insertion, this corresponds to an application of the fixed-length BZ algorithm with a sorted sequence of length $i$. By taking derivative with respect to $\alpha$ for the term $\frac{1-p}{2(1-\alpha)}+\frac{p}{2\alpha}$ in~\eqref{eqn:bz_error}, we can see that the right-hand side of~\eqref{eqn:bz_error} takes minimum value $(n-1)(\tfrac12+\sqrt{p(1-p})^{\td{m}}$ when $\alpha$ is set to $\frac{\sqrt{p}}{\sqrt{p}+\sqrt{1-p}}$. Therefore, by setting $\alpha$ to the optimal value $\frac{\sqrt{p}}{\sqrt{p}+\sqrt{1-p}}$, we can minimize the overall error probability of the noisy sorting code. However, the noisy sorting code does not have the knowledge of the crossover probability $p$. Thus, the proposed noisy sorting code first computes an estimation, denoted $\hat{p}$, for the crossover probability through making $n$ queries on a same pair of elements.
More specifically, we compare the two elements $\theta_1$ and $\theta_2$ for $n$ times. Let $t$ denote the number of responses $\theta_1<\theta_2$ out of the $n$ queries. We set $\hp=\frac{\min(t,n-t)}{n}$ as an \emph{unbiased} estimation for $p$. Then we run each step of BZ algorithm with input parameter $\alpha=\frac{\sqrt{\hp}}{\sqrt{\hp}+\sqrt{1-\hp}}$. Algorithm~\ref{alg:proposed-sorting} shows the pseudocode of the proposed fixed-length noisy sorting code. Notably, this estimation process spends $n$ queries, while the later insertion steps spend in total $\Theta(n\log n)$ queries. Therefore, the number of queries spent on estimation does not affect the dominating term in the total number of queries of the algorithm. Moreover, we show that with high probability these $n$ queries provide a sufficiently accurate estimate $\hp$ for later insertion steps.


To prove Theorem~\ref{thm:achievability}, we set parameter $m$ in Algorithm~\ref{alg:proposed-sorting} to $\lceil\frac{n\log n}{R}\rceil$, where $R$ is a constant satisfying $R<\tfrac12\log(1/g(p))$. Notice that the total number of queries made by Algorithm~\ref{alg:proposed-sorting} is at most $\lceil\frac{n\log n}{R}\rceil+n$, and it follows that the rate of Algorithm~\ref{alg:proposed-sorting} is $R$. 
The proof of Theorem~\ref{thm:achievability} is then completed by the following lemma, which bounds the error probability of Algorithm~\ref{alg:proposed-sorting}.
\begin{lemma}
    \label{lem:fl-sorting-error}
    Suppose $R$ is a constant satisfying $R<\tfrac12\log(1/g(p))$. Then the error probability of Algorithm~\ref{alg:proposed-sorting} with input parameter $m=\lceil\frac{n\log n}{R}\rceil$ is at most $o(1)$.
\end{lemma}
We defer the proof of Lemma~\ref{lem:fl-sorting-error} to Appendix~\ref{app:pf-fl-sorting-error}.


\begin{algorithm}[t]
\caption{Fixed-length Noisy Sorting Code} \label{alg:proposed-sorting}
\DontPrintSemicolon
\SetKwInOut{Input}{Input}
\SetKwInOut{Output}{Output}
\Input{Unordered elements $\theta_1,\ldots,\theta_n$, number of queries for insertions $m$}
\Output{Permutation estimate $\hat{\pi}$.}
\nl Compare elements $\theta_1$ and $\theta_2$ for $n$ times. Let $t$ denote the number of responses $\theta_1<\theta_2$.\\
\nl $\hat{p}\gets \frac{\min(t,n-t)}{n}$.\\
\nl Initialize $\sigma$ as a $1\times 1$ array with element $n$.\\
\For{$i=1,\ldots, n-1$}{
\nl    $\hat{l} \gets \texttt{BZ}\left((\theta_{\sigma(1)},\ldots,\theta_{\sigma(i)}), \: \theta_{n-i}, \: \floor{\frac{m}{n-1}}, \: \frac{\sqrt{\hat{p}}}{\sqrt{\hat{p}}+\sqrt{1-\hat{p}}}\right)$.\\
\nl    $\sigma \gets \big( \sigma(1:\hat{l}-1), \:\, n-i, \:\, \sigma(\hat{l}:\text{end}) \big)$.
}
\nl $\hat{\pi} \gets \sigma$.\tcp*{Treat the $1\times n$ array $\sigma$ as a permutation}
\nl \Return{$\hat{\pi}$}
\end{algorithm}

\subsection{Achievability: Variable-length Noisy Sorting Code}
\label{sec:proof-vl-achievability}
In this section, we prove Theorem~\ref{thm:vl-achievability} by describing a noisy sorting code that can achieve any sorting rate $R<\frac12 (1-H(p))$. The code is similar to the one introduced in Section~\ref{sec:proof-achievability}. The major difference is that each noisy searching step is performed by a variable-length version of the BZ algorithm. We describe the difference between the variable-length BZ algorithm and the fixed-length version in the following section.

\vspace{.5em}
\emph{\textbf{Variable-length BZ algorithm:}}
The variable-length BZ algorithm is based on the same idea of the fixed-length BZ algorithm. There are three differences between these two versions. 
\begin{enumerate}
    \item At each iteration, the variable-length BZ algorithm chooses the index $j^*$ of the query point $\td{\theta}_{j^*}$ by a different principle than~\eqref{eqn:jstar}:
    \begin{equation}\label{eqn:jstar-vl}
    j^{*}(k) = \begin{cases}
    j(k)-1 &\: \text{with probability } \pi_1,\\
    j(k) &\: \text{with probability } \pi_2\triangleq1-\pi_1,
    \end{cases}
    \end{equation}
    where 
    \begin{equation}\label{eqn:pi1}
    \pi_1\triangleq\frac{\rho(\tau_2(k))-\rho(-\tau_2(k))}{\rho(\tau_1(k))-\rho(-\tau_1(k))+\rho(\tau_2(k))-\rho(-\tau_2(k))}
    \end{equation}
    with 
    \begin{equation}\label{eqn:rho}
    \rho(x)\triangleq(1-\hp)\log (1+(1-2\hp)x)+\hp\log (1-(1-2\hp)x)
    \end{equation}
    for $|x|\le 1$, and
    \begin{equation} \label{eqn:bias}
    \tau_1(k)\triangleq 2\sum_{i=j(k)}^{n}q_{k-1}(i) -1\quad  \text{and}\quad \tau_2(k)\triangleq 2\sum_{i=1}^{j(k)}q_{k-1}(i) -1.
\end{equation}
    \item The $\alpha$ value for pseudo-posterior update in~\eqref{eqn:posterior} is chosen to be the estimated crossover probability $\hp$.
    \item The variable-length algorithm does not post a preset $m$ number of queries. Instead, it takes a target error probability $P_e$ as input and runs for a random number $\td{M}$ of queries. For each $i\in [n]$ and $k\ge 0$, define 
    \begin{equation}\label{eq:Z_def}
        Z_k(i)\triangleq \log\frac{q_k(i)}{1-q_k(i)}.
    \end{equation}
    The algorithms terminates at $\td{M}$th iteration, where $\td{M}$ is the smallest integer such that $Z_{\td{M}}(i)\ge-\log(P_e)$ for some $i\in[n]$.
\end{enumerate}
A pseudocode of the variable-length BZ algorithm is given in Algorithm~\ref{alg:vl-bz}.
\begin{algorithm}[t]
\caption{Variable-length Burnashev--Zigangirov \texttt{VLBZ} Algorithm } \label{alg:vl-bz}
\SetKwInOut{Input}{Input}
\SetKwInOut{Output}{Output}
\Input{Ordered sequence $(\td{\theta}_1,\ldots,\td{\theta}_{n-1})$, target $\td{\theta}$, tolerated error probability $P_e$, estimated crossover probability $\hp$. }
\Output{Index $\hat{l}$ of the sub-interval containing $\td{\theta}$.}
\nl $k\gets 0$.\\
\nl $q_0(i) \gets 1/n$ for each $i \in [n]$.\\
\nl Calculate $Z_0(i)$ for each $i \in [n]$ as in~\eqref{eq:Z_def}.\\
\While{$\max_{i}Z_k(i)<-\log(P_e)$} {
\nl $k\gets k+1$\\
\nl     Find $j(k)$ using inequalities in (\ref{eqn:j(k)}).\\
\nl     Set $j^{*}(k)$ according to~\eqref{eqn:jstar-vl}.\\
\nl     Set $\td{U}_k \gets \td{\theta}_{j^{*}(k)}$.\\
\nl     Receive noisy response $\td{Y}_{k}$.\\
\nl     Update $q_k(i)$ for each $i\in [n]$ as in (\ref{eqn:posterior}).\\
\nl     Update $Z_k(i)$ for each $i\in [n]$ as in (\ref{eq:Z_def}).
}
\nl $\hat{l} \gets \arg\max_i \, q_{k}(i)$.\\
\nl \Return{$\hat{l}$}
\end{algorithm}
For variable-length noisy searching algorithms, we assume the true position $l^*$ follows a uniform prior. For the clarity of notation, we use random variable $L^*$ to denote uniformly distributed true position. 

In the following two lemmas, we bound the error probability and the expected number of queries of the variable-length BZ algorithm respectively. 
\begin{lemma}
\label{thm:error-vlbz}
    Suppose $p\le\hp<1/2$. The error probability of the variable-length BZ algorithm is at most $P_e$.
\end{lemma}
\begin{remark}
   If $\hp=p$, the output of variable-length BZ algorithm exactly follows the MAP estimation rule, and it follows that the error probability is at most $P_e$. In Lemma~\ref{thm:error-vlbz}, we generalize this result to the case of $\hp\ge p$. The detailed proof is deferred to Appendix~\ref{app:pf-error-vlbz}.
\end{remark}

\begin{lemma}\label{thm:var-BZ-length}
Let $\gamma=\hp-p$. Suppose $\gamma=o(1)$.
Then the number of queries $\td{M}$ in the variable-length BZ noisy searching scheme satisfies that
\begin{equation}
    \label{eq:var-BZ-length-cond}
    \mathsf{E}[\td{M}\cond L^*=i]\le\frac{\log n-\log P_e+\log\frac{1-\hp}{\hp}}{1-H(p)}(1+O(\gamma))
\end{equation}
for any $i\in[n]$, which further implies that the overall expected number of queries $\mathsf{E}[\td{M}]$ satisfies the same inequality
\begin{equation}
    \label{eq:var-BZ-length-overall}
    \mathsf{E}[\td{M}]\le\frac{\log n-\log P_e+\log\frac{1-\hp}{\hp}}{1-H(p)}(1+O(\gamma)).
\end{equation}
\end{lemma}
\begin{remark}
    An upper bound for the expected number of queries of the variable-length BZ algorithm is first established in~\cite{Burnashev1974}. It is shown that $\mathsf{E}[\td{M}]\le\frac{\log n-\log P_e+\log\frac{1-p}{p}}{1-H(p)}$ in the special case of $\hp=p$. In Lemma~\ref{thm:var-BZ-length}, we extend this result to allow a deviation of $\gamma=o(1)$ from $\hp$ to $p$, and show that a similar upper bound still holds with an additional multiplicative factor of $1+O(\gamma)$. In the original proof of~\cite{Burnashev1974}, a key step is to provide a lower bound for the \emph{expected improvement} of the posterior probability of the correct interval. In our proof of Lemma~\ref{thm:var-BZ-length}, we show that a similar bound with an additional multiplicative factor of $1+O(\gamma)$ holds for the expected improvement of the \emph{pseudo-posterior} of the correct interval under the assumption that $\gamma=\hp-p=o(1)$. This leads to the upper bound for $\E[\td{M}]$ with the same additional multiplicative factor. We defer the detailed proof of Lemma~\ref{thm:var-BZ-length} to Appendix~\ref{app:pf-var-BZ-length}.
\end{remark}

\vspace{.5em}
\emph{\textbf{Proposed Variable-length Noisy Sorting code:}}
The proposed variable-length noisy sorting code calls the variable-length BZ algorithm at each insertion step. Similarly to Algorithm~\ref{alg:proposed-sorting}, the algorithm first spends $n$ queries to estimate the flipping probability $\hp$. However, in this algorithm, we apply a \emph{biased} estimation $\hp=\frac{\min(t,n-t)}{n\left(1-\sqrt{\tfrac{\log n}{n}}\right)}$ instead of the \emph{unbiased} estimation $\hp=\frac{\min(t,n-t)}{n}$ in Algorithm~\ref{alg:proposed-sorting}. This is because when we call the variable-length BZ algorithm, the input parameter is set to $\hp$, and the bound for the error probability of the variable-length BZ algorithm presented in Lemma~\ref{thm:error-vlbz} requires the assumption that $\hp\ge p$. Setting the estimator in this biased manner ensures that $\hp\ge p$ with high probability. The pseudocode is described in Algorithm~\ref{alg:proposed-vlsorting}.
\begin{algorithm}[htbp]
\caption{Variable-Length Noisy Sorting Code} \label{alg:proposed-vlsorting}
\DontPrintSemicolon
\SetKwInOut{Input}{Input}
\SetKwInOut{Output}{Output}
\Input{Unordered elements $\theta_1,\ldots,\theta_n$, tolerated error probability $P_e$}
\Output{Permutation estimate $\hat{\pi}$.}
\nl Compare elements $\theta_1$ and $\theta_2$ for $n$ times. Let $t$ denote the number of responses $\theta_1<\theta_2$.\\
\nl $\hat{p}\gets \frac{\min(t,n-t)}{n\left(1-\sqrt{\tfrac{\log n}{n}}\right)}$.\\
\nl Initialize $\sigma$ as a $1\times 1$ array with element $n$.\\
\nl  \For{$i=1,\ldots, n-1$}{
\nl \If{\# comparisons made so far  $\geq n(\log n)^2$}{
\nl \textbf{break}
}
\nl    $\hat{l} \gets \texttt{VLBZ}\big((\theta_{\sigma(1)},\ldots,\theta_{\sigma(i)}), \: \theta_{n-i}, \: \frac{P_e}{n}, \: \hp\big)$.\\
\nl    $\sigma \gets \big( \sigma(1:\hat{l}-1), \:\, n-i, \:\, \sigma(\hat{l}:\text{end}) \big)$.
}
\nl $\hat{\pi} \gets \sigma$.\tcp*{Treat the $1\times n$ array $\sigma$ as a permutation}
\nl \Return{$\hat{\pi}$}
\end{algorithm}

To prove Theorem~\ref{thm:vl-achievability}, we the input parameter $P_e$ of Algorithm~\ref{alg:proposed-vlsorting} to $\frac{1}{\log n}$, and bound the expected number of queries and the error probability in the following two lemmas respectively. The proof of Theorem~\ref{thm:vl-achievability} follows readily from the two lemmas. The proofs of the two lemmas are presented in Appendices~\ref{app:pf-vl-sorting-queries} and~\ref{app:pf-vl-sorting-error} respectively.
\begin{lemma}\label{lem:vl-sorting-queries}
    Suppose the input parameter $P_e$ of Algorithm~\ref{alg:proposed-vlsorting} is set to $\frac{1}{\log n}$. Then the number of queries $M$ spent by the algorithm satisfies
    \begin{equation*}
        \lim_{n\to\infty}\frac{n\log n}{\E[M]}\ge \frac{1-H(p)}{2}.
    \end{equation*}
\end{lemma}
\begin{lemma}\label{lem:vl-sorting-error}
    Suppose the input parameter $P_e$ of Algorithm~\ref{alg:proposed-vlsorting} is set to $\frac{1}{\log n}$. Then the error probability of the algorithm is at most $o(1)$.
\end{lemma}

\subsection{Comparison with Existing Algorithms}
\label{sec:knownAlg}
In this section, we compare the rates of the proposed algorithms to that of~\cite{ren2019sample}, which has the best known sorting rate when the noise probability $p$ is unknown. Note that the algorithm in~\cite{ren2019sample} represents an extension of the algorithm in~\cite{Feige1994} to accommodate scenarios involving unknown and query-dependent noise levels. However, when we narrow our focus to the case of a fixed $p$, the sorting rate of the algorithm in~\cite{ren2019sample} consistently remains upper-bounded by the sorting rate of the algorithm in~\cite{Feige1994}. Therefore, in this section, we shift our attention to the algorithm proposed in~\cite{Feige1994}, and obtain its sorting rate through a refined analysis. This rate subsequently serves as an upper bound for the sorting rate of the algorithm in \cite{ren2019sample}. It is shown (see Fig. ~\ref{fig:comparison}) that the two proposed noisy sorting codes described in Section~\ref{sec:proof-achievability} and Section~\ref{sec:proof-vl-achievability} both strictly improve the sorting rate of the algorithm in~\cite{ren2019sample} for all $p \in (0,1/2)$.

\vspace{.5em}
\emph{\textbf{Noisy Searching using Comparison Trees:}}
Consider the noisy searching problem outlined in Section~\ref{sec:noisy-search}.  
The algorithm in~\cite{Feige1994} defines a binary tree $T$, where each node in the tree is labeled with a pair of indices $(i,j)$ representing interval $(\td{\theta}_i,\td{\theta}_j]$, for $0\leq i<j\leq n$. 
The root of the tree is labeled with the index pair $(0,n)$. The rest of the tree is defined recursively as follows: For each node with a label $(i,j)$ such that $j-i\ge 2$, two children are generated such that the left child has the label $(i,\ceil{\frac{i+j}{2}})$, and the right child has the label $(\ceil{\frac{i+j}{2}},j)$. By the above process, we build a binary tree $T$ that contains $n$ leaves, each representing one of the possible intervals in $\{I_i\}$ (recall the definition of sub-intervals $I_i = (\td{\theta}_{i-1},\td{\theta}_i]$ for $i\in [n-1]$, and $I_n = (\td{\theta}_{n-1},\td{\theta}_n)$). To allow searching with noisy queries, the algorithm defines an extended tree $T^{*}$ by attaching a length-$s$ chain of nodes to each leaf node in $T$ having the same label as the leaf node. Fig.~\ref{fig:NST} shows an illustration of the trees $T$ and $T^{*}$ for the case when $n=6$ and $s = 3$. 

\begin{figure}[tbp]
    \centering
    \def\svgscale{0.77}
    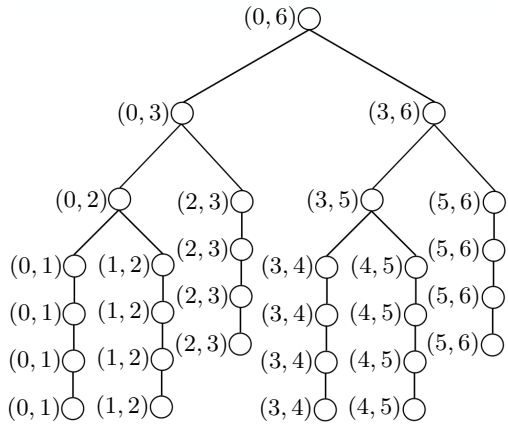
    \caption{Illustration of the extended tree $T^*$ for $n=6$ and $s=3$.}
    \label{fig:NST}
\end{figure}

The algorithm performs a random walk on the nodes of $T^{*}$ starting from the root node, as follows. 
Let $(i,j)$ be the label of the current node reached by the random walk.
First, we verify if $\theta$ belongs to the corresponding interval. To do this, the queries $\td{\theta}_i$ and $\td{\theta}_j$ are submitted $\beta$ times each, for some integer $\beta$ taken as input to the algorithm. If the majority vote of the responses indicates that either $\td{\theta} \leq \td{\theta}_i$ or $\td{\theta} > \td{\theta}_j$, we claim that the verification fails and the algorithm steps back to the parent node of $(i,j)$. In case the verification succeeds, if $j-i \ge 2$, i.e., the current node is neither a leaf node in $T$ nor an extended node, the algorithm checks if $\td{\theta}$ is larger or smaller than the median element of the current interval.  This is done by submitting, $\beta$ times, the query $\td{\theta}_{\ceil{\frac{i+j}{2}}}$. The algorithm proceeds in the direction indicated by the majority vote of the responses. Otherwise when $j-i = 1$, the algorithm proceeds to the only child of the current node. Algorithm~\ref{alg:NST} shows the pseudocode of this comparison tree algorithm.

\begin{algorithm}[t]
	\caption{Noisy Searching Tree \texttt{NST} Algorithm} \label{alg:NST}
	\SetKwInOut{Input}{Input}
	\SetKwInOut{Output}{Output}
	\DontPrintSemicolon
	\Input{Ordered sequence $(\td{\theta}_1,\ldots,\td{\theta}_{n-1})$, target $\td{\theta}$, number of queries $\td{m}$, integer $\beta$.}
	\Output{Index $\hat{l}$ of the sub-interval containing $\td{\theta}$.}
	\nl $i \gets 0$, $j \gets n$, $s\gets \ceil{\frac{\td{m}}{3\beta}}$. \\
	\nl Construct $T^{*}$ as described in text.\\
	\nl \For{$k=1,\ldots,s$} {
		\nl $b \gets 0$.\tcp*[r]{failure bit} 
		\nl \uIf{$i \neq 0$} {
			\nl Submit query $\td{U}_k \gets \td{\theta}_i$ for $\beta$ times.\\
			\nl Get noisy responses $\td{Y}_1, \ldots, \td{Y}_\beta$.\\
			\nl \lIf{$\sum_{i=1}^{\beta}\td{Y}_i < \ceil{\frac{\beta}{2}}$} {
				$b \gets 1$.
			}
		}
		\nl \uIf{$j \neq n$} {
			\nl Submit query $\td{U}_k \gets \td{\theta}_j$ for $\beta$ times.\\
			\nl Get noisy responses $\td{Y}_1, \ldots, \td{Y}_\beta$.\\
			\nl \lIf{$\sum_{i=1}^{\beta}\td{Y}_i \geq \ceil{\frac{\beta}{2}}$} {
				$b \gets 1$.
			}
		}
		\nl \uIf{$b = 1$} {
			\nl $(i,j) \gets \texttt{parent}\big((i,j)\big)$.
		}
		\nl \uElseIf{$j-i\geq 2$} {
			\nl Submit query $\td{U}_k \gets \td{\theta}_{\ceil{\frac{i+j}{2}}}$ for $\beta$ times.\\
			\nl Get noisy responses $\td{Y}_1, \ldots, \td{Y}_\beta$.\\
			\nl \uIf{$\sum_{i=1}^{\beta}\td{Y}_i < \ceil{\frac{\beta}{2}}$} {
				\nl $(i,j) \gets \texttt{left\_child}\big((i,j)\big)$.
			}
			\nl \uElse{
				\nl $(i,j) \gets \texttt{right\_child}\big((i,j)\big)$.
			}
		}
		\nl \uElse {
		    \tcp{extended chain}
			\nl $(i,j) \gets \texttt{child}\big((i,j)\big)$.
		}
	}
	\nl $\hat{l} \gets i +1$.\\
	\nl \Return{$\hat{l}$}
\end{algorithm}
\vspace{.5em}
\emph{\textbf{Insertion Sort with Noisy Comparison Tree:}}
As before, a noisy sorting code can be constructed by integrating the noisy searching algorithm (Algorithm~\ref{alg:NST}) as part of an insertion-based algorithm, as discussed in Section~\ref{sec:proof-achievability}. This can be done by simply replacing Line 2 of Algorithm~\ref{alg:proposed-sorting} by the following function call,\\[.2em]
\centerline{
$
\hat{l} \gets \texttt{NST}\big((\theta_{\sigma(1)},\ldots,\theta_{\sigma(i-1)}), \: \theta_{n-i}, \: \left\lfloor\tfrac{m}{n-1}\right\rfloor, \: \beta\big),
$}\\[.2em]
where $\beta$ is a parameter inputted to the noisy sorting algorithm, and corresponds to the number of query repetitions done at each node of the tree.

\vspace{.5em}
\emph{\textbf{Rate of Algorithm:}}
In this section, we apply a refined error analysis to the algorithm described above. The achievable rate of the algorithm is stated in the following theorem.
\begin{theorem}[Refined rate of the algorithm in~\cite{Feige1994}]
\label{thm:achievability2}
For $0 <\delta <\frac12$, sorting rate $R$ is achievable if 
\begin{align}
    3\beta(p,\delta)R &< 1- 2\delta,\label{eq:thm3_cond1}\\
    3\beta(p,\delta)R &< \frac{1}{\ln 2} D\left(\frac12-\frac{3\beta(p,\delta)R}{2}\bigg|\bigg|\delta\right),\label{eq:thm3_cond2}
\end{align}
where $\beta(p,\delta)$ is the smallest integer $\beta$ such that
\[
 1-\left(1-\sum_{i=\ceil{\beta/2}}^\beta\binom{\beta}{i}p^i(1-p)^{\beta-i}\right)^3 \le \delta.
\]
This rate is achieved using the algorithm in~\cite{Feige1994} based on noisy searching and insertion sort, where a pseudocode of the noisy searching algorithm is given in Algorithm~\ref{alg:NST}.
\end{theorem}

To prove Theorem~\ref{thm:achievability2}, we first consider the noisy searching problem and upper bound the probability of error of Algorithm~\ref{alg:NST} in the following lemma. The proof of the lemma is deferred to Appendix~\ref{app:pf-noisy-searching}

\begin{lemma}\label{lem:noisy_searching}
    Let $0<\delta<\frac12$ and positive integer $\tilde{m}$ be such that 
    \begin{equation} \label{eq:lem1_cond1}
    \frac12-\frac{\log n}{2s}>\delta,
    \end{equation}
    where $s = \ceil{\frac{\tilde{m}}{3\beta(p,\delta)}}$ and $\beta(p,\delta)$ is the smallest integer $\beta$ such that
\[
 1-\left(1-\sum_{i=\ceil{\beta/2}}^\beta\binom{\beta}{i}p^i(1-p)^{\beta-i}\right)^3 \le \delta.
\]
Then Algorithm~\ref{alg:NST} with input parameters $\beta=\beta(p,\delta)$ and $\tilde{m}$ outputs the correct interval for $\tilde{\theta}$ with probability at least
    \[
    1-\exp\left(-sD\left(\frac12-\frac{\log n}{2s}\bigg|\bigg|\delta\right)\right).
    \]
\end{lemma}

\begin{proof}[Proof of Theorem~\ref{thm:achievability2}]
Let $m = \ceil{\frac{n\log n}{R}}$ be the total number of queries for noisy sorting, where the sorting rate $R$ satisfies conditions~(\ref{eq:thm3_cond1}) and (\ref{eq:thm3_cond2}) of Theorem~\ref{thm:achievability2}. Let $\td{m} = \floor{\frac{m}{n-1}}$ be the number of queries used in each insertion, and let $s = \ceil{\frac{\tilde{m}}{3\beta(p,\delta)}}$.\footnote{For simplicity of presentation, we assume in the proof that $\frac{\log n}{R}$ is an integer.} Let $P_e(i)$ denote the probability of error of the $i$th insertion. In order to upper bound $P_e(i)$ using Lemma~\ref{lem:noisy_searching}, we first verify the condition of Lemma~\ref{lem:noisy_searching}. Notice that
\[
s = \left\lceil\frac{\tilde{m}}{3\beta(p,\delta)}\right\rceil=\left\lceil\frac{\floor{\frac{\ceil{\frac{n\log n}{R}}}{n-1}}}{3\beta(p,\delta)}\right\rceil\geq \frac{\log n}{3R\beta(p,\delta)}.
\]
It follows that
\[
\frac{1}{2} - \frac{\log n}{2s} \geq \frac{1}{2} - \frac{3R\beta(p,\delta)}{2} > \delta,
\]
where the last inequality holds since $R$ satisfies condition~(\ref{eq:thm3_cond1}). Therefore, the error probability of the noisy sorting algorithm of~\cite{Feige1994} can be bounded by
\begin{align*}
    &\P\{\hat{\Pi}\neq\pi\} \leq \sum_{i=2}^n P_e(i) \\
    &\stackrel{(a)}{\le}\sum_{i=2}^n \exp\left\{-\frac{\tilde{m}}{3\beta(p,\delta)}D\left(\frac{1}{2}-\frac{3\beta(p,\delta)\log i}{2\tilde{m}}\Big|\Big|\delta\right)\right\}\\
    &\stackrel{(b)}{\le}n\exp\left\{-\frac{\tilde{m}}{3\beta(p,\delta)}D\left(\frac{1}{2}-\frac{3\beta(p,\delta)\log n}{2\tilde{m}}\Big|\Big|\delta\right)\right\}\\
    &\leq \exp\left\{ \log n \left( \ln 2 - \frac{1}{3R\beta(p,\delta)}D\left(\frac{1}{2}-\frac{3R\beta(p,\delta)}{2} \Big|\Big|\delta \right)  \right)   \right\}
\end{align*}
where $(a)$ follows by applying Lemma~\ref{lem:noisy_searching}, and $(b)$ follows since $D(x||y)$ is an increasing function in $x$ whenever $0<y<x<\frac{1}{2}$. Hence, $\P\{\hat{\Pi}=\pi\}$ goes to zero as $n\rightarrow\infty$ if the sorting rate $R$ satisfies condition~\eqref{eq:thm3_cond2}. 
\end{proof}

Note that the analysis in~\cite{Feige1994} considers the special case of $\delta = \frac{1}{3}$, whereas our analysis considers any $0<\delta < 1/2$ such that~(\ref{eq:thm3_cond1}) and (\ref{eq:thm3_cond2}) hold. It turns out that such a generalization of the analysis significantly improves the achievable sorting rates of the algorithm. Finally, we point out that the achievable sorting rates shown in Fig.~\ref{fig:comparison} for the existing algorithm correspond to the best choice of $\delta$ for each $p$.

\section{Proof of Converse Results}\label{sec:proof-conv}
In this section, we prove the two converse results Theorem~\ref{thm:vl-converse} and Theorem~\ref{thm:insert-sort-conv}. A key building block of both  proofs is a converse theorem on the \emph{channel coding with feedback problem} in~\cite{Burn1976}. In the following, we first formally define the channel coding with feedback problem and state the converse result.
\subsection{Coding Problem on a BSC$(p)$ Channel with Noiseless Feedback}
Suppose the transmitter has a set of $n$ possible messages $\{\psi_1,\ldots, \psi_n\}$ 
and he wants to transmit one of them to the receiver through a BSC$(p)$ channel. At each iteration $k$, the transmitter sends a bit $\td{X}_k$ and the receiver receives $\td{Y}_k=\td{X}_k\oplus \td{Z}_k$, where $\td{Z}_k \sim \mathrm{Bern}(p)$ is a sequence of i.i.d. random variables. Furthermore, after receiving each transmitted bit, the receiver can communicate with the transmitter through a noiseless discrete feedback channel. The alphabet of the feedback channel is assumed to be arbitrarily large but finite. A variable-length transmission scheme is specified as follows: To start the transmission, the receiver sends a signal $\td{W}_0$ to the transmitter through the feedback channel. At iteration $k$, the bit $\td{X}_k$ sent by the transmitter is decided by a causal protocol $\td{X}_k=\td{f}_k(\Psi^*, \td{X}_1^{k-1},\td{W}_0^{k-1})$, where $\Psi^*$ is the true message transmitter wants to send. Then the receiver receives $\td{Y}_k=\td{X}_k\oplus \td{Z}_k$ and generates $\td{\lambda}_k = \td{g}_k(\td{Y}_1^{k}, \td{W}_0^{k-1})$, where $\td{\lambda}_k\in [n+1]$. In the case when $\td{\lambda}_k\in [n]$, the transmission is terminated and the receiver outputs $\hat{\Psi}=\psi_{\td{\lambda}_k}$. Otherwise, the receiver decides to continue transmission, and generate $\td{W}_k=\td{h}_k(\td{Y}_1^{k},\td{W}_0^{k-1})$ to send to the transmitter through the feedback channel. The total number of bits sent by the transmitter is denoted as $\td{M}$. The performance of a transmission scheme is evaluated by two quantities. The probability of error $\td{P}_e\triangleq \P(\hat{\Psi}\neq \Psi^*)$ and the expected number of queries $\E[\td{M}]$. Assuming the true message $\Psi^*$ has a uniform prior over the message set $\{\psi_1,\ldots,\psi_n\}$, we have the following converse bound for any transmission scheme. 
\begin{theorem}[Theorem 1 in~\cite{Burn1976}]\label{thm:burn_conv}
For any variable-length transmission scheme on the $\mathrm{BSC}(p)$ channel with feedback, the expected number of bits sent by the transmitter satisfies that
\begin{equation}
        \E[\td{M}]> \frac{\log n}{1-H(p)}+\frac{\log\frac{1}{\td{P}_e}}{D(p||1-p)}-\frac{\log\log \frac{n}{\td{P}_e}}{D(p||1-p)}-\frac{\td{P}_e\log n}{1-H(p)}+O(1).
\end{equation}
\end{theorem}
In the following two sections, we use Theorem~\ref{thm:burn_conv} to prove Theorem~\ref{thm:vl-converse} and Theorem~\ref{thm:insert-sort-conv} by connecting the noisy sorting problem and the noisy searching problem to the channel coding with feedback problem.
\subsection{Converse: Variable-length Noisy Sorting Codes}
To prove Theorem~\ref{thm:vl-converse}, we show that for each noisy sorting code that sorts $n$ elements with $M$ queries and average case error probability $\P(\hat{\Pi}\neq \Pi)=P_e$ under the uniform prior for the true permutation $\Pi$, it induces a coding scheme for the channel coding with feedback problem with $n!$ possible input messages. The induced coding scheme also has $M$ number of bits sent and error probability $P_e$. Then Theorem~\ref{thm:vl-converse} follows from Theorem~\ref{thm:burn_conv}.

Firstly, we view the $n!$ message $\psi_1,\ldots,\psi_{n!}$ as the $n!$ possible permutations $\pi_1,\ldots,\pi_{n!}$ on $n$ elements in an arbitrary order. Let $\pi^*$ denote the permutation corresponding to the true message $\Psi^*$. For a noisy sorting code with query strategy $\{f_k\}_{k=1}^\infty$ associated with stopping time $M$ and estimator $\hat{\Pi}=\hat{\pi}(\td{Y}^M,U^M,V^M)$, the induced coding scheme is described as follows: To start the transmission, the receiver sends $\td{W}_0=f_1=(U_1,V_1)$, which is the first pair of elements compared by the noisy sorting code, to the transmitter. For the simplicity of notation, let $(U_k,V_k)\triangleq (\theta_{k_1},\theta_{k_2})$. The at each iteration $k=1,2,\ldots,$ the transmitter sends $\td{X}_k=\mathbbm{1}_{\pi^*(k_1)<\pi^*(k_2)}$, which is a function of the message $\pi^*$ and $(U_k,V_k)$, through the BSC$(p)$ channel. The receiver receives the channel output $\td{Y}_k=\td{X}_k\oplus \td{Z}_k$ and generates 
 $\td{\lambda}_k$ based on the stopping time $M$ of the noisy sorting code. If the stopping criteria is not satisfied, then $\td{\lambda}_k$ is set to $n!+1$, i.e., the transmission continues and the receiver generates the feedback symbol $\td{W}_k=f_{k+1}(\td{Y}_{1}^{k},\td{W}_0^{k-1})=(U_{k+1},V_{k+1})$ to send to the transmitter. Otherwise, the transmission terminates and $\td{\lambda}_k$ is set to the message corresponding to the permutation $\hat{\Pi} = \hat{\pi}(\td{Y}^k,U^k,V^k)$, where $\hat{\Pi}$ is the estimator of the noisy sorting code.
In other words, at each iteration, the receiver sends the queried pair of the noisy sorting code to the transmitter through the noiseless feedback channel, and the transmitter sends the bit indicating the comparison between the queried pair through the BSC$(p)$ channel. Upon receiving the bit, the receiver decides to output an estimation or continue transmission based on the stopping criteria of the noisy sorting code. One can check that the induced coding scheme has the same error probability $P_e$ and the number of transmitted bits is exactly $M$. By Theorem~\ref{thm:burn_conv}, we have that the expected number of queries made by the noisy sorting code satisfies that
\begin{equation} \label{eqn:converse_proof}
        \E[M]> \frac{\log (n!)}{1-H(p)}+\frac{\log\frac{1}{P_e}}{D(p||1-p)}-\frac{\log\log \frac{n!}{P_e}}{D(p||1-p)}-\frac{P_e\log (n!)}{1-H(p)}+O(1),
\end{equation}
and Theorem~\ref{thm:vl-converse} follows because the worst case error probability of a noisy sorting code is always at least the average case error probability.

\begin{remark}
    The converse bound presented in Theorem~\ref{thm:vl-converse} remains valid when the crossover probability $p$ is known. 
    However, it has been established in~\cite{gu2023optimal} that in this scenario, $C(p)=\frac{1}{\frac{1}{1-p}+\frac{1}{D(p||1-p)}}$, revealing that the bound in Theorem~\ref{thm:vl-converse} is not fully tight. This discrepancy arises because the bound is derived from treating the noisy sorting problem as a specialized case of the channel coding with feedback problem, and employing Burnashev's converse for channel coding.
    As demonstrated in the proof, we conceptualize noisy sorting algorithms as specialized instances of feedback coding schemes, where the feedback message is constrained to pairwise comparisons. More specifically, if we can set the transferred message $\tilde{X}_k=\mathbbm{1}_{\{\pi^*\in\mathcal{A}\}}$ for \emph{any} set of permutations $\mathcal{A}$ at each iteration, then the Burnashev's converse for channel coding with feedback would result in a tight bound. However, pairwise comparison restricts the type of query set $\mathcal{A}$ to be the set of permutations such that $\pi^*(i) < \pi^*(j)$. 
    This observation suggests the necessity of considering the constraint on the query set $\mathcal{A}$ to attain a tight bound based on Burnashev's methodology.
    Regrettably, with this constraint, analyzing the posterior distribution at each iteration turns out to be a formidable challenge, and whether Burnashev's methodology can lead to a tight converse remains unclear.
\end{remark}
\subsection{Converse: Insertion-based Noisy Sorting Codes}
\label{sec:proof-ins-conv}
The key components of the insertion-based noisy sorting codes are the noisy searching algorithms applied at each insertion step. Thus, to show the converse for the family of insertion-based noisy sorting codes, we first state and prove a converse result for all noisy searching algorithms. Recall the setup for the noisy searching problem introduced in Section~\ref{sec:noisy-search}. We have the following lower bound for the expected number of queries.
\begin{theorem}\label{thm:noisy-search-conv}
    Suppose the true position $l^*$ follows a uniform prior. For the clarity of notation, we use $L^*$ to denote the random true position. For any variable-length noisy searching scheme with error probability $P_e\triangleq \P(\hat{L}\neq L^*)$, its expected number of queries $\E[\td{M}]$ satisfies that
\begin{equation}
    \E[\td{M}]> \frac{\log n}{1-H(p)}+\frac{\log\frac{1}{\td{P}_e}}{D(p||1-p)}-\frac{\log\log \frac{n}{\td{P}_e}}{D(p||1-p)}-\frac{\td{P}_e\log n}{1-H(p)}+O(1).
\end{equation}
\end{theorem}
\begin{proof}[Proof of Theorem~\ref{thm:noisy-search-conv}]
    Similarly to the proof of Theorem~\ref{thm:vl-converse}, we prove Theorem~\ref{thm:noisy-search-conv} by drawing a connection to the channel coding with feedback problem. 

    To see the connection, we first observe that the $n$ possible messages $\psi_1,\ldots,\psi_n$ in the channel coding problem can be viewed as the $n$ sub-intervals $I_1,\ldots,I_n$ respectively. Denote the mid point of the sub-interval corresponding to the true message $\Psi^*$ by $\theta^*$. For a noisy searching scheme with query strategy $\{\td{f}_k\}_{k=1}^\infty$ associated with stopping time $\td{M}$ and estimator $\hat{L}=\hat{l}(\td{Y}^{\td{M}},\td{U}^{\td{M}})$, the induced coding scheme is described as follows: To start the transmission, the receiver sends $\td{W}_0=\tilde{f}_1=\td{U}_1$, which is the first query in the query scheme, to the transmitter. Then at each iteration $k=1,2,\ldots$, the transmitter sends $\td{X}_k=\mathbbm{1}_{\{\td{U}_k<\theta^*\}}$ through the BSC$(p)$ channel. The receiver receives the channel output $\td{Y}_k=\td{X}_k\oplus \td{Z}_k$ and generates $\td{\lambda}_k$ based on the stopping time $\td{M}$ of the noisy searching algorithm. If the stopping criteria is not satisfied, then $\td{\lambda}_k$ is set to $n+1$, i.e., the transmission continues and the receiver generates the feedback symbol $\td{W}_k=\td{f}_{k+1}(\td{Y}_{1}^{k},\td{W}_0^{k-1})=\td{U}_{k+1}$ to send to the transmitter. Otherwise, the transmission terminates and $\td{\lambda}_k$ is set to message $\hat{L}=\hat{l}(\td{Y}^k,\td{U}^k)$, where $\hat{L}$ is the estimator of the noisy searching algorithm. In other words, at each iteration, the receiver sends the query point to the transmitter through the feedback channel, and the transmitter sends the bit indicating whether the message is on the left of the query point through the BSC$(p)$ channel. Upon receiving the bit, the receiver decides to output an estimation or continue transmission based on the stopping rule. One can check that the induced coding scheme has exactly the same error probability $\td{P}_e$ and the number of bits sent is exactly $\td{M}$. Since each noisy searching algorithm induces a coding scheme, Theorem~\ref{thm:noisy-search-conv} follows from Theorem~\ref{thm:burn_conv}.
\end{proof}
With Theorem~\ref{thm:noisy-search-conv}, we are now ready to prove Theorem~\ref{thm:insert-sort-conv}. 
For $i\in[n-1]$, let $\mathcal{E}_i$ denote the event that estimation at $i$th insertion step is incorrect. We assume without loss of generality that $\pi$ follows a uniform prior. For the clarity of notation, we use $\Pi$ to denote the uniformly distributed true permutation. By the similar argument as in the proof of Theorem~\ref{thm:vl-achievability}, we know that 
\[
\P(\hat{\Pi}=\Pi)=\prod_{k=1}^{n-1}\left(1-\P(\mathcal{E}_k \cond \cap_{i=1}^{k-1}\mathcal{E}_i^c)\right),
\]
where $\P(\mathcal{E}_i|\cap_{k=1}^{i-1}\mathcal{E}_k^c)$ represents the error probability of the noisy searching algorithm called at $i$th insertion step given a uniform prior for the correct position. Notice that 
\begin{align*}
    \P(\hat{\Pi}=\Pi)&=
    \prod_{k=1}^{n-1}\left(1-\P(\mathcal{E}_k \cond \cap_{i=1}^{k-1}\mathcal{E}_i^c)\right)\\
    &=\exp{\left(\sum_{i=1}^{n-1}\ln\left(1-\P(\mathcal{E}_i|\cap_{k=1}^{i-1}\mathcal{E}_k^c)\right)\right)}\\
    &\stackrel{(a)}{\le} \exp{\left(-\sum_{i=1}^{n-1}\P(\mathcal{E}_i|\cap_{k=1}^{i-1}\mathcal{E}_k^c)\right)},
\end{align*}
where (a) follows because $\ln x \leq x-1$.
We can see that a necessary condition for $\lim_{n\rightarrow\infty}\P(\hat{\Pi}=\Pi)=1$ is that at least $n-w$ terms in $\P(\mathcal{E}_1),\P(\mathcal{E}_2|\mathcal{E}_1^c),\ldots,\P(\mathcal{E}_{n-1}|\mathcal{E}_1^c\cap\cdots\cap\mathcal{E}_{n-2}^c)$ are $o(1/n)$, where $w=o(n)$. Assume otherwise, there exists $w=\Omega(n)$ terms that are $\Omega(1/n)$. Then the sum of these $w$ terms must be in the order $\Omega(1)$, which implies that $\lim_{n\rightarrow\infty}\P(\hat{\Pi}=\Pi)\neq 1$. We denote the iteration indices of these $n-w$ terms by $j_1,\ldots,j_{n-w}$. In other words, for every $j\in\{j_1,\ldots,j_{n-w}\}$, we have $\P(\mathcal{E}_{j}|\mathcal{E}_1^c\cap\cdots\cap\mathcal{E}_{j-1}^c)=o(1/n)$.

Let $M_k$ denote the number of queries in the $k$th insertion step. By Theorem~\ref{thm:noisy-search-conv}, we have that for every $l\in [n-w]$, 
\begin{equation}
    \E[M_{j_l}]>\frac{\log(j_l+1)}{1-H(p)}+\frac{\log n}{D(p||1-p)}+O(\log\log n).
\end{equation}
Furthermore, we can lower bound the total expected number of queries in the noisy sorting code as
\begin{align}
    \E[M]&\ge\E\left[\sum_{l=1}^{n-w}M_{j_l}\right]\nonumber\\
    &=\sum_{l=1}^{n-w}\E[M_{j_l}]\nonumber\\
    &>\sum_{l=1}^{n-w}\left[\frac{\log(j_l+1)}{1-H(p)}+\frac{\log n}{D(p||1-p)}+O(\log\log n)\right]\nonumber\\
    &\stackrel{(a)}{\ge}\sum_{k=1}^{n-w}\left[\frac{\log(k+1)}{1-H(p)}+\frac{\log n}{D(p||1-p)}+O(\log\log n)\right]\nonumber\\
    &=\frac{\log(n-w+1)!}{1-H(p)}+\frac{(n-w)\log n}{D(p||1-p)}+O(n\log\log n)\nonumber\\
    &\stackrel{(b)}{=}n\log n\left(\frac{1}{1-H(p)}+\frac{1}{D(p||1-p)}\right)+O(n\log n),\label{eq:exp_lower}
\end{align}
where (a) follows by the fact that $j_1,\ldots,j_{n-w}$ are distinct indices, so $\sum_{l=1}^{n-w}\log(j_l+1)\ge \sum_{k=1}^{n-w}\log(k+1)$, and (b) follows by Stirling's approximation and the fact that $w=o(n)$. By equation~\eqref{eq:exp_lower}, we have $ R_\mathrm{ins}=\frac{n\log n}{\E[M]}<\frac{1}{\frac{1}{1-H(p)}+\frac{1}{D(p||1-p)}}$, which completes the proof.

\section{Concluding Remarks}\label{sec:conclusion}
In this paper, we proposed a new framework to study the noisy sorting problem. We introduced the noisy sorting capacity, which is defined as the maximum ratio $\frac{n\log n}{\E[M]}$ such that the ranking of the $n$ elements can be correctly found with probability going to one by an algorithm with $\E[M]$ expected number of queries. We provide two lower bounds on the noisy sorting capacity, which strictly improve the best-known achievable rate. 
Notably, the proposed algorithm does not require the knowledge of the noise level $p$.
We also provide two converse results for the problem, one for general variable-length noisy sorting algorithms and the other for the family of insertion-based noisy sorting algorithms. 

There are many intriguing open problems that remain. Firstly, we can see from Fig.~\ref{fig:comparison} that there exists a gap between our converse and achievability results. Therefore, a tighter characterization of the noisy sorting capacity is of interest. We conjecture that an algorithm based on the idea of posterior matching in~\cite{Shayevitz2011} could potentially achieve a higher sorting rate. The main idea of the algorithm is as follows. The algorithm assumes a uniform prior of the space $\mathcal{S}_n$ of all possible $n!$ permutations. At each iteration, the algorithm considers each of the $\binom{n}{2}$ possible pair of items for comparison. Notice that each pair of items $\{\theta_i,\theta_j\}$ naturally partitions the space $\mathcal{S}_n$ into two equal-sized parts where the first part includes all the permutations such that $\theta_i<\theta_j$ and the other part includes all the permutations such that $\theta_i>\theta_j$. The algorithm computes the sum of posterior over these two parts for each pair and chooses the pair with the smallest difference of the two sums as the query to be made at this iteration. The iteration is then completed by updating all the posteriors by the response to the query. The algorithms terminate when at least one of the posterior reaches a preset threshold. According to empirical experiments, the algorithms achieve vanishing error probability with fewer queries than the proposed algorithms in this paper. Secondly, we have considered the noisy sorting capacity under the assumption that the error probability for each query is fixed and independent of the queried pair. This leads to a natural extension of the error model where the error probability between each pair could be distinct. Finally, we considered adaptive noisy sorting codes in this work. Another open question is how the capacity changes when the noisy sorting codes are restricted to non-adaptive algorithms.


\section*{Acknowledgment}
The authors would like to thank Young-Han Kim and Wei Yu for stimulating discussions on the problem. This work was supported in part by the NSERC Discovery Grant No. RGPIN-2019-05448, in part by the NSERC Collaborative Research and Development Grant CRDPJ 54367619 and in part by NSF Grants IIS-1901252 and CCF-190949.

\bibliographystyle{IEEEtran}
\bibliography{bibliography}
\newpage
\appendix
\section{Proof of Lemmas}\label{app:pf-lemma}
\subsection{Proof of Lemma~\ref{lem:fl-sorting-error}}\label{app:pf-fl-sorting-error}
To bound the error probability of Algorithm~\ref{alg:proposed-sorting}, we first bound the difference between $\hp$ and $p$. Without loss of generality, we assume $\theta_1>\theta_2$. Therefore, we have $t\sim\mathrm{Binom}(n,p)$. By the Chernoff bound, we have 
$$\P\left(np\left(1-\sqrt{\frac{\log n}{n}}\right)\le t\le np\left(1+\sqrt{\frac{\log n}{n}}\right)\right)\ge 1- 2n^{-p/3}.$$
Let $\mathcal{A}\triangleq\left\{p\left(1-\sqrt{\frac{\log n}{n}}\right)\le\hp\le p\left(1+\sqrt{\tfrac{\log n}{n}}\right)\right\}$.
It follows that 
$$\P\left(\mathcal{A}\right)\ge 1-2n^{-p/3}.$$
On event $\mathcal{A}$, we have $$\frac{\sqrt{\hp}}{\sqrt{\hp}+\sqrt{1-\hp}}=\frac{\sqrt{p}}{\sqrt{p}+\sqrt{1-p}}\left(1+O\left(\sqrt{\tfrac{\log n}{n}}\right)\right).$$
Moreover, because $\frac{\sqrt{p}}{\sqrt{p}+\sqrt{1-p}}>p$ and $\frac{\sqrt{p}}{\sqrt{p}+\sqrt{1-p}}-p=\Theta(1)$ under the assumption that $p\in (0,1/2)$, we know that $\frac{\sqrt{\hp}}{\sqrt{\hp}+\sqrt{1-\hp}}>p$ on event $\mathcal{A}$.
Note that the function $\frac{1-p}{2(1-\alpha)}+\frac{p}{2\alpha}$ is Lipschitz continuous at its saddle point $\frac{\sqrt{\hp}}{\sqrt{\hp}+\sqrt{1-\hp}}$. Therefore, we have 
\begin{equation}
\label{eq:hat-error}
    \frac{1-p}{2(1-\ha)}+\frac{p}{2\ha}=\left(\tfrac12+\sqrt{p(1-p)}\right)\left(1+O\left(\sqrt{\tfrac{\log n}{n}}\right)\right).
\end{equation}
Let $P_e(i)$ denote the probability of error of the $i$th insertion. 
Let $\hat{\Pi}$ denote the output of Algorithm~\ref{alg:proposed-sorting}.
Then, the error probability of the proposed noisy sorting code can be bounded as
\begin{align*}
    \P\{\hat{\Pi} \neq \pi|\mathcal{A}\} &\leq \sum_{i=1}^{n-1} P_e(i)\\
    &\overset{(a)}{\leq} \sum_{i=1}^{n-1} i\left(\left(\tfrac12+\sqrt{p(1-p)}\right)\left(1+O\left(\sqrt{\tfrac{\log n}{n}}\right)\right)\right)^{\floor{\frac{m}{n-1}}}\\
    &\overset{(b)}{\leq} \frac{n(n-1)}{2} \left(\left(\tfrac12+\sqrt{p(1-p)}\right)\left(1+O\left(\sqrt{\tfrac{\log n}{n}}\right)\right)\right)^{\frac{\log n}{R}-1}\\
    &\leq n^2 \exp{\left\{-\frac{\ln n}{R}\left(\log(1/g(p))+\log \left(1+O\left(\sqrt{\tfrac{\log n}{n}}\right)\right)\right)\right\}}\\
    &=\leq n^2 \exp{\left\{-\frac{\ln n}{R}\left(1+O\left(\sqrt{\tfrac{\log n}{n}}\right)\right)\log(1/g(p))\right\}}
\end{align*}
where $g(p) = \frac{1}{2} + \sqrt{p(1-p)}$, $(a)$ follows from Lemma~\ref{thm:fl-bz} and equation~\eqref{eq:hat-error}, and $(b)$ follows by setting $m = \ceil{\frac{n\log n}{R}}$. By setting 
\[
R < \frac{1}{2}\log(1/g(p)),
\]
we have $\P\{\hat{\Pi} \neq \pi|\mathcal{A}\}\le n^{-\Theta(1)}$. Finally, the proof is completed by realizing
\[
\P\{\hat{\Pi} \neq \pi\}\le \P\{\hat{\Pi} \neq \pi|\mathcal{A}\}+\P(\mathcal{A}^c)=o(1).
\]
\subsection{Proof of Lemma~\ref{thm:error-vlbz}}\label{app:pf-error-vlbz}
For each $i\in [n]$ and $k\ge 1$, let $q^*_k(i)=\P\{\td{\theta} \in I_i \cond \td{Y}_1,\ldots,\td{Y}_{k},\td{U}_1,\ldots,\td{U}_{k}\}$ denote the posterior probability of $\td{\theta}\in I_i$ after $k$ rounds of queries. We can see that $q_k(i)=q^*_k(i)$ in the special case of $\hat{p}=p$. Let $l_k\triangleq\arg\max_i \, q_{k}(i)$. It suffices to show that $q^*_k(l_k)\ge q_k(l_k)$ for any $k\ge 1$. Fix an arbitrary $k\ge 1$. For $i\in[n]$, let $a_i$ denote the number of query-response pairs such that $\td{y}_{k} = \mathbbm{1}_{\{j^{*}(k)\leq i\}}$. By the update rule~\eqref{eqn:posterior}, we know that for any $i,j\in[n]$, we have 
    \begin{equation*}
        \frac{q_k(i)}{q_k(j)}=\frac{(1-\hp)^{a_i}\hp^{k-a_i}}{(1-\hp)^{a_j}\hp^{k-a_j}}.
    \end{equation*}
    Since $\hp<1/2$, we know that $a_{l_k}=\max_{i\in[n]}a_i$. Furthermore, because~\eqref{eqn:posterior} guarantees that $\sum_{i=1}^nq_k(i)=1$, we have
    \begin{equation*}
        q_k(i)=\frac{(1-\hp)^{a_i}\hp^{k-a_i}}{\sum_{j=1}^n(1-\hp)^{a_j}\hp^{k-a_j}}.
    \end{equation*}
    Similarly, by the Bayes' Rule, we have
    \begin{equation*}
        q_k^*(i)=\frac{(1-p)^{a_i}p^{k-a_i}}{\sum_{j=1}^n(1-p)^{a_j}p^{k-a_j}}.
    \end{equation*}    
    Let $r_1=\frac{\hp}{p}\ge 1$ and $r_2=\frac{1-\hp}{1-p}\le 1$. We have
    \begin{align*}
        \frac{q_k(l_k)}{q^*_k(l_k)}&=r_1^{k-a_{l_k}}r_2^{a_{l_k}}\frac{\sum_{j=1}^n(1-p)^{a_j}p^{k-a_j}}{\sum_{j=1}^n(1-\hp)^{a_j}\hp^{k-a_j}}\le 1,
    \end{align*}
    where the last inequality follows because $$\frac{(1-p)^{a_j}p^{k-a_j}}{(1-\hp)^{a_j}\hp^{k-a_j}}=\frac{1}{r_1^{k-a_j}r_2^{a_j}}\le \frac{1}{r_1^{k-a_{l_k}}r_2^{a_{l_k}}},$$
    for each $j\in [n]$, and hence $\frac{\sum_{j=1}^n(1-p)^{a_j}p^{k-a_j}}{\sum_{j=1}^n(1-\hp)^{a_j}\hp^{k-a_j}}\le  \frac{1}{r_1^{k-a_{l_k}}r_2^{a_{l_k}}}$. This shows that $q^*_k(l_k)\ge q_k(l_k)$, and completes the proof. 
    \subsection{Proof Lemma~\ref{thm:var-BZ-length}}\label{app:pf-var-BZ-length}
    Recall that $\td{M}$ is the first time step such that there exists $i\in[n]$ with $Z_{\td{M}}(i)\ge -\log P_e$. Let $\td{M}'$ denote the first time step such that $Z_{\td{M}'}(l^*)\ge -\log P_e$. By our definitions, we know that $\td{M}\le \td{M}'$. So it satisfies to give an upper bound for $\E[\td{M}']$. 

Recall that we defined $Z_{k}(l^*)=\log \frac{q_{k}(l^*)}{1-q_{k}(l^*)}$, which characterizes the ratio between the pseudo-posterior of the correct interval and that of the sum of the rest. For the simplicity of notation, we define $\bar{q}_k\triangleq (q_k(1),\ldots,q_k(n))$. In the following, we consider the \emph{expected improvement} $\E[Z_{k}(l^*)\cond \bar{q}_{k-1}]-Z_{k-1}(l^*)$ for each time step $t$. We prove that $$\E[Z_{k}(l^*)\cond \bar{q}_{k-1}]-Z_{k-1}(l^*)\ge (1-H(p))(1+O(\hp-p))$$ at each time step $k$. Then the proof of the lemma is completed by a martingale argument.
To prove the lower bound for the expected improvement, we fix a positive integer $k\ge 1$. Recall that we defined the median of the pseudo-posterior $j(k)$ as
\begin{equation*}
    \sum_{i=1}^{j(k)-1}q_{k-1}(i) \leq \frac{1}{2} \quad \text{and} \quad \sum_{i=1}^{j(k)}q_{k-1}(i) > \frac{1}{2}.
\end{equation*}
We consider three different cases: (1) $l^*=j(k)$, (2) $l^*>j(k)$, and (3) $l^*<j(k)$. Let $h_1$, $h_2$ and $h_3$ denote the expected improvement for these three cases respectively.
We first consider the case of $l^*=j(k)$. By taking the randomness over the choice of the query point and the noise in the observation model, we can compute
\begin{align*}
    h_1&=\pi_1\left((1-p)\log\frac{2(1-q_{k-1}(l^*))(1-\hp)}{1+(1-2\hp)\tau_1(k-1)-q_{k-1}(l^*)2(1-\hp)}+p\log\frac{2(1-q_{k-1}(l^*))\hp}{1+(2\hp-1)\tau_1(k-1)-q_{k-1}(l^*)2\hp}\right)\\
    &+\pi_2\left((1-p)\log\frac{2(1-q_{k-1}(l^*))(1-\hp)}{1+(1-2\hp)\tau_2(k-1)-q_{k-1}(l^*)2(1-\hp)}+p\log\frac{2(1-q_{k-1}(l^*))\hp}{1+(2\hp-1)\tau_2(k-1)-q_{k-1}(l^*)2\hp}\right).
\end{align*}
Let 
\begin{align*}
    h_1'&\triangleq \pi_1\left((1-p)\log\frac{2(1-q_{k-1}(l^*))(1-p)}{1+(1-2p)\tau_1(k-1)-q_{k-1}(l^*)2(1-p)}+p\log\frac{2(1-q_{k-1}(l^*))p}{1+(2p-1)\tau_1(k-1)-q_{k-1}(l^*)2p}\right)\\
    &+\pi_2\left((1-p)\log\frac{2(1-q_{k-1}(l^*))(1-p)}{1+(1-2p)\tau_2(k-1)-q_{k-1}(l^*)2(1-p)}+p\log\frac{2(1-q_{k-1}(l^*))p}{1+(2p-1)\tau_2(k-1)-q_{k-1}(l^*)2p}\right).
\end{align*}
We are going to see that $h_1=h_1'+O(\gamma)$. Let $\lambda=1-q_{k-1}(l^*)$. Notice that
\begin{align}
    &|1+(1-2\hp)\tau_1(k-1)-q_{k-1}(l^*)2(1-\hp)-(1+(1-2p)\tau_1(k-1)-q_{k-1}(l^*)2(1-p))|\nonumber\\
    &=|-2\gamma\tau_1(k-1)+2\gamma q_{k-1}(l^*)|\nonumber\\
    &\stackrel{(a)}{\le}2\gamma\lambda,\label{eq:vlbz-length-1}
\end{align}
where (a) follows because $1-2\lambda=q_{k-1}(l^*)-(1-q_{k-1}(l^*))\le\tau_1(k-1)\le 1$ and $q_{k-1}(l^*)=1-\lambda$. Moreover, notice that
\begin{align}
    &1+(1-2p)\tau_1(k-1)-q_{k-1}(l^*)2(1-p)\nonumber\\
    &\ge 1+(1-2p)(2q_{k-1}(l^*)-1)-q_{k-1}(l^*)2(1-p)\nonumber\\
    &=2p-2pq_{k-1}(l^*)=2p\lambda.\label{eq:vlbz-length-2}
\end{align}
Equations \eqref{eq:vlbz-length-1} and \eqref{eq:vlbz-length-2} together implies that
$$1+(1-2\hp)\tau_1(k-1)-q_{k-1}(l^*)2(1-\hp)=(1+O(\gamma))(1+(1-2p)\tau_1(k-1)-q_{k-1}(l^*)2(1-p)).$$
It follows that
\begin{align}
    \log\frac{2(1-q_{k-1}(l^*))(1-\hp)}{1+(1-2\hp)\tau_1(k-1)-q_{k-1}(l^*)2(1-\hp)}&=\log\frac{2(1-q_{k-1}(l^*))(1-p)}{1+(1-2p)\tau_1(k-1)-q_{k-1}(l^*)2(1-p)}+\log(1+O(\gamma))\nonumber\\
    &=\log\frac{2(1-q_{k-1}(l^*))(1-p)}{1+(1-2p)\tau_1(k-1)-q_{k-1}(l^*)2(1-p)}+O(\gamma).\label{eq:case1-1}
\end{align}
Similarly, we can show that 
\begin{align}
    \log\frac{2(1-q_{k-1}(l^*))\hp}{1+(2\hp-1)\tau_1(k-1)-q_{k-1}(l^*)2\hp}&=\log\frac{2(1-q_{k-1}(l^*))p}{1+(2p-1)\tau_1(k-1)-q_{k-1}(l^*)2p}+O(\gamma);\label{eq:case1-2}\\
    \log\frac{2(1-q_{k-1}(l^*))(1-\hp)}{1+(1-2\hp)\tau_2(k-1)-q_{k-1}(l^*)2(1-\hp)}&=
    \log\frac{2(1-q_{k-1}(l^*))(1-p)}{1+(1-2p)\tau_2(k-1)-q_{k-1}(l^*)2(1-p)}+O(\gamma);\label{eq:case1-3}\\
    \log\frac{2(1-q_{k-1}(l^*))\hp}{1+(2\hp-1)\tau_2(k-1)-q_{k-1}(l^*)2\hp}&=\log\frac{2(1-q_{k-1}(l^*))p}{1+(2p-1)\tau_2(k-1)-q_{k-1}(l^*)2p}+O(\gamma).\label{eq:case1-4}
\end{align}
Equations~\eqref{eq:case1-1}-\eqref{eq:case1-4} together imply that $h_1=h_1'+O(\gamma)$.

In the cases of $l^*>j(k)$ and $l^*<j(k)$, we can go through a similar algebraic procedure to show that $h_2=h_2'+O(\gamma)$ and $h_3=h_3'+O(\gamma)$, where
\begin{align*}
    h_2'&\triangleq \pi_1\left((1-p)\log\frac{2(1-q_{k-1}(l^*))(1-p)}{1+(1-2p)\tau_1(k-1)-q_{k-1}(l^*)2(1-p)}+p\log\frac{2(1-q_{k-1}(l^*))p}{1+(2p-1)\tau_1(k-1)-q_{k-1}(l^*)2p}\right)\\
    &+\pi_2\left((1-p)\log\frac{2(1-q_{k-1}(l^*))(1-p)}{1+(2p-1)\tau_2(k-1)-q_{k-1}(l^*)2(1-p)}+p\log\frac{2(1-q_{k-1}(l^*))p}{1+(1-2p)\tau_2(k-1)-q_{k-1}(l^*)2p}\right), 
\end{align*}
and
\begin{align*}
    h_3'&\triangleq \pi_1\left((1-p)\log\frac{2(1-q_{k-1}(l^*))(1-p)}{1+(2p-1)\tau_1(k-1)-q_{k-1}(l^*)2(1-p)}+p\log\frac{2(1-q_{k-1}(l^*))p}{1+(1-2p)\tau_1(k-1)-q_{k-1}(l^*)2p}\right)\\
    &+\pi_2\left((1-p)\log\frac{2(1-q_{k-1}(l^*))(1-p)}{1+(1-2p)\tau_2(k-1)-q_{k-1}(l^*)2(1-p)}+p\log\frac{2(1-q_{k-1}(l^*))p}{1+(2p-1)\tau_2(k-1)-q_{k-1}(l^*)2p}\right).
\end{align*}

Recall that we defined 
\begin{equation*}
    \rho(x)\triangleq (1-\hp)\log (1+(1-2\hp)x)+\hp\log (1-(1-2\hp)x).
    \end{equation*}
By some algebra, we get 
\begin{equation}
    h_1'=1-H(p)-\pi_1 \rho\left(\frac{\tau_1(k-1)-q_{k-1}(l^*)}{1-q_{k-1}(l^*)}\right)-\pi_2 \rho\left(\frac{\tau_2(k-1)-q_{k-1}(l^*)}{1-q_{k-1}(l^*)}\right),\label{eq:h1_eq}
\end{equation}
\begin{equation}
    h_2'=1-H(p)-\pi_1 \rho\left(\frac{\tau_1(k-1)-q_{k-1}(l^*)}{1-q_{k-1}(l^*)}\right)-\pi_2 \rho\left(-\frac{\tau_2(k-1)+q_{k-1}(l^*)}{1-q_{k-1}(l^*)}\right),\label{eq:h2_eq}
\end{equation}
and
\begin{equation}
    h_3'=1-H(p)-\pi_1 \rho\left(-\frac{\tau_1(k-1)+q_{k-1}(l^*)}{1-q_{k-1}(l^*)}\right)-\pi_2 \rho\left(\frac{\tau_2(k-1)-q_{k-1}(l^*)}{1-q_{k-1}(l^*)}\right).\label{eq:h3_eq}
\end{equation}
We are going to show that $h_1'\ge 1-H(p)$, $h_2'\ge 1-H(p)$ and $h_3'\ge 1-H(p)$.
Consider first cases (2) and (3). Notice that $\frac{y-x}{1-x}\le y$ and $-\frac{y+x}{1-x}\le -y$ for any $0<x<1$ and $-1 \leq y \leq 1$. By the fact that $\rho(x)$ is non-decreasing for $-1\le x\le 1$, we have 
\begin{align}
    h_2'&\ge 1-H(p)-\pi_1 \rho\left(\tau_1(k-1)\right)-\pi_2 \rho\left(-\tau_2(k-1)\right);\label{eq:h2_bound}\\
    h_3'&\ge 1-H(p)-\pi_1 \rho\left(-\tau_1(k-1)\right)-\pi_2 \rho\left(\tau_2(k-1)\right).\label{eq:h3_bound}
\end{align}
Recall that we set $\pi_1=\frac{\rho(\tau_2)-\rho(-\tau_2)}{\rho(\tau_1)-\rho(-\tau_1)+\rho(\tau_2)-\rho(-\tau_2)}$, which implies that 
\begin{equation}
\label{eq:h2_h3_eq}
\pi_1 \rho\left(\tau_1(k-1)\right)+\pi_2 \rho\left(-\tau_2(k-1)\right)=\pi_1 \rho\left(-\tau_1(k-1)\right)+\pi_2 \rho\left(\tau_2(k-1)\right).
\end{equation}
By equations~\eqref{eq:h2_bound},~\eqref{eq:h3_bound} and~\eqref{eq:h2_h3_eq}, we have
\begin{align*}
    h_2'&\ge 1-H(p)-\frac12(\pi_1 \rho\left(\tau_1(k-1)\right)+\pi_2 \rho\left(-\tau_2(k-1)\right)+\pi_1 \rho\left(-\tau_1(k-1)\right)+\pi_2 \rho\left(\tau_2(k-1)\right))\\
    &\ge 1-H(p)-\frac12(\pi_1(\rho(\tau_1(k-1))+\rho(-\tau_1(k-1)))+\pi_2(\rho(\tau_2(k-1))+\rho(-\tau_2(k-1))))\\
    &\stackrel{(a)}{\ge} 1-H(p)-\frac12\left(\max_{-1\le x\le 1}\rho(x)+\rho(-x)\right),
\end{align*}
where (a) follows because $\pi_1+\pi_2=1$. But $\frac12\left(\max_{-1\le x\le 1}\rho(x)+\rho(-x)\right)\le 0$, because
\begin{align*}
    \rho(x)+\rho(-x)&=(1-p)\log(1+(1-2p)x)+p\log(1-(1-2p)x)+(1-p)\log(1-(1-2p)x)+p\log(1+(1-2p)x)\\
    &=\log(1+(1-2p)x)+\log(1-(1-2p)x)\\
    &=\log(1-(1-2p)^2x^2)\le 0.
\end{align*}
This shows that $h_2'\ge 1-H(p)$, and we can get $h_3'\ge 1-H(p)$ in the same way.

Now, we consider case (1). Using equation~\eqref{eq:h1_eq} and the fact that $\tau_1+\tau_2=2q_{k-1}(l^*)$, we have
\begin{equation}
\label{eq:h1_eq_2}
    h_1'=1-H(p)-\pi_1 \rho\left(\frac{\tau_1(k-1)-q_{k-1}(l^*)}{1-q_{k-1}(l^*)}\right)-\pi_2 \rho\left(-\frac{\tau_1(k-1)-q_{k-1}(l^*)}{1-q_{k-1}(l^*)}\right).
\end{equation}
For simplicity, we denote the term $\frac{\tau_1(k-1)-q_{k-1}(l^*)}{1-q_{k-1}(l^*)}$ by $T$. Here we want to show that 
\begin{equation}\label{eq:case1_cond}
\pi_1\rho(T)+\pi_2\rho(-T)\le 0.
\end{equation}
Towards that goal, we consider two sub-cases $\tau_1\le\tau_2$ and $\tau_1\ge\tau_2$. In the first case, we have $\pi_1\ge \pi_2$, because $\rho(x)-\rho(-x)$ is non-decreasing for $0<x<1$. Let $\pi_1=\frac12+\epsilon$ and $\pi_2=\frac12-\epsilon$ for some positive constant $\epsilon$. Moreover, we have $T\le 0$ because $\tau_1(k-1)+\tau_2(k-1)=2q_{k-1}(l^*)$ and $\tau_1\le\tau_2$. It follows that 
\begin{align*}
    &\pi_1\rho(T)+\pi_2\rho(-T)\\
    &=(\tfrac12+\epsilon)((1-p)\log(1+(1-2p)T)+p\log(1-(1-2p)T))\\&+(\tfrac12-\epsilon)((1-p)\log(1-(1-2p)T)+p\log(1+(1-2p)T))\\
    &=\tfrac12\log(1-(1-2p)^2T^2)+\epsilon((1-2p)\log(1+(1-2p)T)+(2p-1)\log(1-(1-2p)T)).
\end{align*}
Our claim~\eqref{eq:case1_cond} holds because $\log(1-(1-2p)^2T^2)\le 0$, $(1-2p)\log(1+(1-2p)T)\le 0$ and $(2p-1)\log(1-(1-2p)T)\le 0$. The proof for the second sub-case follows similarly.

Therefore, in all three cases, we can lower bound the expected improvement $\E[Z_{k}(l^*)\cond \Bar{q}_{k-1}]-Z_{k-1}(l^*)\ge 1-H(p)+O(\gamma)=(1-H(p))(1+O(\gamma))$, where the last equality follows because $1-H(p)=\Theta(1)$. Now, we state a lemma to relate the expected improvement to the total expected number of queries. 
\begin{lemma}\label{lem:submart_stop}
Let $X_0,X_1,X_2,\ldots$ be a discrete-time submartingale with respect to $\sigma$-algebra filtration $\mathcal{F}_0,\mathcal{F}_1,\mathcal{F}_2,\ldots$. For $k=0,1,2,\ldots$, suppose
\begin{equation}
    \E[X_{k+1}\cond\mathcal{F}_k]\ge X_{k}+C
\end{equation}
for some $C>0$, and
\begin{equation}
    |X_{k+1}-X_{k}|\le A
\end{equation}
for some finite constant $A$.
Define a stopping-time $\tau$ by $\tau\triangleq\min\{k:X_k\ge B\}$. Then the expectation of $\tau$ satisfies inequality
\begin{equation}
    \E[\tau\cond\mathcal{F}_0]\le\frac{B-X_0+A}{C}.
\end{equation}
\end{lemma}

We first use Lemma~\ref{lem:submart_stop} to complete the proof of Lemma~\ref{thm:var-BZ-length} and then present the proof of the Lemma in the next subsection. In our case, we have that $\E[Z_{k}(l^*)\cond \Bar{q}_{k-1}]-Z_{k-1}(l^*)\ge (1-H(\hp))(1+O(\hp-p))$. Clearly, $Z_0(l^*),Z_1(l^*),Z_2(l^*),\ldots$ form a submartingale with respect to the sequence of $\sigma$-algebras induced by $\Bar{q}_0,\Bar{q}_1,\Bar{q}_2,\ldots$. To apply Lemma~\ref{lem:submart_stop}, we still need to bound the difference $|Z_{k-1}(l^*)-Z_{k}(l^*)|$ for each $k$. We claim that $$|Z_{k-1}(l^*)-Z_{k}(l^*)|\le \log\frac{1-\hp}{\hp}.$$
To see the claim, we consider two different cases: (1) $l^*>j^*(k)$ and (2) $l^*\le j^*(k)$, where $j^*(k)$ is randomly chosen from $j(k)-1$ and $j(k)$ as in~\eqref{eqn:jstar-vl}.

In the first case, recall that $\tau(k) = 2\sum_{i=1}^{j^{*}(k)}q_{k-1}(i) - 1.$ By the update rule, we know that the pseudo-posterior of the correct interval increases when $\td{Y}_k=1$, and it is updated as $q_k(l^*)=q_{k-1}(l^*)\frac{2(1-\hp)}{1+(2\hp-1)\tau(k)}$. When $\td{Y}_k=0$, the pseudo-posterior of the correct interval decreases, and it is updated as $q_k(l^*)=q_{k-1}(l^*)\frac{2\hp}{1+(1-2\hp)\tau(k)}$. If $q_k(l^*)=q_{k-1}(l^*)\frac{2(1-\hp)}{1+(2\hp-1)\tau(k)}$, we have
\begin{align}
    Z_{k}(l^*)-Z_{k-1}(l^*)-\log\frac{1-\hp}{\hp}&=\log\frac{q_{k-1}(l^*)\frac{2(1-\hp)}{1+(2\hp-1)\tau(k)}}{1-q_{k-1}(l^*)\frac{2(1-\hp)}{1+(2\hp-1)\tau(k)}}-\log\frac{q_{k-1}(l^*)}{1-q_{k-1}(l^*)}-\log\frac{1-\hp}{\hp}\nonumber\\
    &=\log\frac{(1-q_{k-1}(l^*))\frac{2(1-\hp)}{1+(2\hp-1)\tau(k)}}{1-q_{k-1}(l^*)\frac{2(1-\hp)}{1+(2\hp-1)\tau(k)}}-\log\frac{1-\hp}{\hp}\nonumber\\
    &\stackrel{(a)}{\le} \log\frac{2(1-\hp)}{1+(2\hp-1)\tau(k)}-\log\frac{1-\hp}{\hp}\nonumber\\
    &=\log\frac{2\hp}{1+(2\hp-1)\tau(k)}\nonumber\\
    &=\log\left(1+\frac{2\hp-(1+(2\hp-1)\tau(k))}{1+(2\hp-1)\tau(k)}\right)\nonumber\\
    &=\log\left(1+\frac{(2\hp-1)(1-\tau(k))}{1+(2\hp-1)\tau(k)}\right)\stackrel{(b)}{\le} 0,\label{eq:incr_case}
\end{align}
where (a) follows since $\frac{y-x}{1-x}\le y$ for any $0<x<1$ and $y\leq 1$, and (b) follows since $2\hp-1<0$ and $1-\tau(k)>0$. 
Similarly, if $q_k(l^*)=q_{k-1}(l^*)\frac{2\hp}{1+(1-2\hp)\tau(k)}$, we have 
\begin{align}
    Z_{k}(l^*)-Z_{k-1}(l^*)+\log\frac{1-\hp}{\hp}&=\log\frac{q_{k-1}(l^*)\frac{2\hp}{1+(1-2\hp)\tau(k)}}{1-q_{k-1}(l^*)\frac{2\hp}{1+(1-2\hp)\tau(k)}}-\log\frac{q_{k-1}(l^*)}{1-q_{k-1}(l^*)}+\log\frac{1-\hp}{\hp}\nonumber\\
    &=\log\frac{(1-q_{k-1}(l^*))\frac{2\hp}{1+(1-2\hp)\tau(k)}}{1-q_{k-1}(l^*)\frac{2\hp}{1+(1-2\hp)\tau(k)}}+\log\frac{1-\hp}{\hp}\nonumber\\
    &=\log\frac{(1-q_{k-1}(l^*))\frac{2(1-\hp)}{1+(1-2\hp)\tau(k)}}{1-q_{k-1}(l^*)\frac{2\hp}{1+(1-2\hp)\tau(k)}}\nonumber\\
    &=\log\left(\frac{2(1-\hp)-2(1-\hp)q_{k-1}(l^*)}{1+(1-2\hp)\tau(k)-2\hp q_{k-1}(l^*)}\right)\nonumber\\
    &=\log\left(1+\frac{2(1-\hp)-2(1-\hp)q_{k-1}(l^*)-1-(1-2\hp)\tau(k)+2\hp q_{k-1}(l^*)}{1+(1-2\hp)\tau(k)-2\hp q_{k-1}(l^*)}\right)\nonumber\\
    &=\log\left(1+\frac{(1-2\hp)(1-\tau(k)-2q_{k-1}(l^*))}{1+(1-2\hp)\tau(k)-2\hp q_{k-1}(l^*)}\right)\nonumber\\
    &\stackrel{(a)}{\ge} 0,\label{eq:decr_case}
\end{align}
where (a) follows since $1-\tau(k)=2\sum_{i=j^*(k)+1}^nq_{k-1}(i)$ and $l^*\ge j^*(k)+1$. By~\eqref{eq:incr_case} and~\eqref{eq:decr_case}, we know that $| Z_{k}(l^*)-Z_{k-1}(l^*)|\le \log \frac{1-\hp}{\hp}$ in this case. The proof for case (2) goes in the same way as case (1) and it is skipped. 

Because $Z_0(1)=\cdots=Z_0(n)=\log \frac{1}{n-1}$, by Lemma~\ref{lem:submart_stop}, we have 
\begin{align*}
    \E[\td{M}\cond l^*=i]&\le \frac{-\log P_e-\log\frac{1}{n-1}+\log\frac{1-\hp}{\hp}}{(1-H(p))(1+O(\gamma))}\\
    &\le \frac{-\log P_e+\log n+\log\frac{1-\hp}{\hp}}{1-H(p)}(1+O(\gamma)),
\end{align*}
which completes the proof.
\subsection{Proof of Lemma~\ref{lem:submart_stop}}\label{app:pf-submart_stop}
For $k=0,1,\ldots$, let $Y_k\triangleq X_k-kC$. We claim that $Y_0,Y_1,\ldots,$ is also a submartingale with respect to $\sigma$-algebra filtration $\mathcal{F}_0,\mathcal{F}_1,\ldots$. This is because
\begin{align*}
    \E[Y_{k+1}\cond\mathcal{F}_k]&=\E\left[X_{k+1}-kC-C\cond\mathcal{F}_k\right]\\
    &\ge X_k+C-kC-C\\
    &=Y_k.
\end{align*}
Because $\tau$ is a stopping time with respect to $\mathcal{F}_0,\mathcal{F}_1,\ldots$, by optional stopping time theorem (see for example~\cite{grimmett2020}), we know that $Y_{k\wedge \tau}$ also forms a submartingale, where $\wedge$ denotes the minimum operator. We have that
\begin{align*}
    Y_0&\stackrel{(a)}{\le} \lim_{k\rightarrow\infty}\E[Y_{k\wedge\tau}\cond\mathcal{F}_0]\\
    &=\lim_{k\rightarrow\infty}\E[X_{k\wedge\tau}-C(k\wedge\tau)\cond\mathcal{F}_0]\\
    &\le \lim_{k\rightarrow\infty}\E[X_{k\wedge\tau}\cond\mathcal{F}_0]-C\E[\tau\cond\mathcal{F}_0]\\
    &\stackrel{(b)}{\le}B+A-C\E[\tau\cond\mathcal{F}_0],
\end{align*}
where (a) follows because $Y_{k\wedge\tau}$ forms a submartingale, and (b) follows because $X_{(k\wedge\tau)-1}\le B$ and $|X_{(k\wedge\tau)-1}-X_{k\wedge\tau}|\le A$. The proof is then completed by realizing that $Y_0=X_0$.

\subsection{Proof of Lemma~\ref{lem:vl-sorting-queries}}\label{app:pf-vl-sorting-queries}
Let $\mathcal{A}\triangleq\left\{p\le\hp\le p\left(1+3\sqrt{\tfrac{\log n}{n}}\right)\right\}$.
From the proof of Lemma~\ref{lem:fl-sorting-error},
we have $$\P\left(np\left(1-\sqrt{\frac{\log n}{n}}\right)\le t\le np\left(1+\sqrt{\frac{\log n}{n}}\right)\right)\ge 1- 2n^{-p/3},$$ which implies that $\P(\mathcal{A}^c)\le 2n^{-p/3}$.
Let $M_\rmins$ denote the number of queries spent on the insertion steps of Algorithm~\ref{alg:proposed-vlsorting}. Let $M_i$ denote the number of queries at $i$th insertion. By the law of total expectation, we have
\[
\E[M_\rmins]=\E[M_\rmins|\mathcal{A}]\P(\mathcal{A})+\E[M_\rmins|\mathcal{A}^c]\P(\mathcal{A}^c)\le \E[M_\rmins|\mathcal{A}]+\E[M_\rmins|\mathcal{A}^c]\P(\mathcal{A}^c).
\]
Because the algorithm always makes at most $n(\log n)^2$ queries, we have $\E[M_\rmins|\mathcal{A}^c]\P(\mathcal{A}^c)=O(n(\log n)^2n^{-p/3})=o(n)$. Furthermore, by Lemma~\ref{thm:var-BZ-length}, we have
\begin{align}
    \E[M_\rmins|\mathcal{A}]&=\sum_{i=1}^{n-1}\E[M_i|\mathcal{A}]\nonumber\\
    &\le \sum_{i=1}^{n-1}\frac{\log(i+1)-\log\td{P}_e(i)+\log\frac{1-\hp}{\hp}}{1-H(p)}(1+O(\gamma))\nonumber\\
    &\le\frac{\log(n!)+n\log n+n\log\log n+n\log\frac{1-\hp}{\hp}}{1-H(p)}(1+O(\gamma))\nonumber\\
    &\stackrel{(a)}{=}\frac{2n\log n+O(n\log\log n)}{1-H(p)},\label{eq:ins-queries}
\end{align}
where (a) follows by Stirling's approximation and $\gamma=O(\sqrt{\log n /n})$. Therefore, we have $\E[M_\rmins]\le\frac{2n\log n+O(n\log\log n)}{1-H(p)}$, and it follows that the total number queries $M$ satisfies $\lim_{n\to\infty}\frac{n\log n}{\E[M]}=\lim_{n\to\infty}\frac{n\log n}{\E[M_\rmins+n]}\ge \frac{1-H(p)}{2}$.

\subsection{Proof of Lemma~\ref{lem:vl-sorting-error}}\label{app:pf-vl-sorting-error}
As we mentioned in Remark~\ref{rem:reduction}, we can assume without loss of generality that $\pi$ follows a uniform prior. In this proof, for the clarity of notation, we use random variable $\Pi$ to denote the uniformly distributed true permutation. Let $\hat{\Pi}$ denote the output of Algorithm~\ref{alg:proposed-vlsorting}.
Notice that $\P(\hat{\Pi}=\Pi)\ge \P(\hat{\Pi}=\Pi,\mathcal{A})=\P(\hat{\Pi}=\Pi|\mathcal{A})\P(\mathcal{A})$. Since we know that $\P(\mathcal{A}^c)\le 2n^{-p/3}$, it suffices to show that $\P(\hat{\Pi}=\Pi|\mathcal{A})=1-o(1)$. To bound $\P(\hat{\Pi}=\Pi|\mathcal{A})$, let us assume there is an auxiliary algorithm, which is identical to Algorithm~\ref{alg:proposed-vlsorting}, 
except that the termination conditions presented on lines 5-6 in Algorithm~\ref{alg:proposed-vlsorting} are removed.
Comparing Algorithm~\ref{alg:proposed-vlsorting} and the auxiliary algorithm, their query strategies and estimation rules are same, while Algorithm~\ref{alg:proposed-vlsorting} has the additional restriction that the total number of queries is at most $n(\log n)^2$. Therefore the probability that Algorithm~\ref{alg:proposed-vlsorting} gives the correct output equals to the probability that the auxiliary algorithm gives the correct output \emph{and} the total number of queries is at most $n(\log n)^2$. Let $\Pi'$ denote the output of the auxiliary algorithm, and let $N_\rmins$ denote the number of queries it spends on the insertions steps. Let $\mathcal{B}\triangleq\{N_\rmins\le n(\log n)^2-n\}$. Because the algorithm spends $n$ queries on estimating $p$, we have
$$\P(\hat{\Pi}=\Pi|\mathcal{A})=\P(\Pi'=\Pi,\mathcal{B}|\mathcal{A})\ge 1-\P(\Pi'\neq\Pi|\mathcal{A})-\P(\mathcal{B}^c|\mathcal{A}),$$
where the last inequality follows from the union bound.
From~\eqref{eq:ins-queries}, we know that $\E[N_\rmins|\mathcal{A}]=\Theta(n\log n)$. Therefore, by Markov's inequality, we have $\P(\mathcal{B}^c|\mathcal{A})=O(1/\log n)$. Now, we move on to bound $\P(\Pi'=\Pi|\mathcal{A})$.
For each insertion step, we have already bounded its error probability under the assumption of a uniform prior for its correct position in Lemma~\ref{thm:error-vlbz}. Here we still need to argue that given the uniform prior on the permutation $\Pi$, the true position for each insertion step follows a uniform prior. Towards that goal, we introduce an alternative representation of a permutation proposed in~\cite{Lehmer1960} known as the Lehmer code, and utilize two of its properties from~\cite{Lehmer1960}.

\begin{definition}
For a permutation $\pi\in\mathcal{S}_n$, the Lehmer code is an $n$-tuple $L(\pi)=(L(\pi)_1,\ldots,L(\pi)_n)$, where $L(\pi)_i=|\{j>i:\pi(j)<\pi(i)\}|$.
\end{definition}
\begin{lemma}
The Lehmer code defines a bijection from the set of all permutations $\mathcal{S}_n$ to the Cartesian product space $([n-1]\cup\{0\})\times([n-2]\cup\{0\})\times\cdots\times \{0\} $.
\end{lemma}
\begin{lemma}\label{lem:Lehmer-indep}
Suppose $\Pi\sim\mathrm{Unif}(\mathcal{S}_n)$. Then the $L(\Pi)_i\sim\mathrm{Unif}\{0,\ldots,n-i\}$ for each $i\in [n]$. Furthermore, the entries $L(\Pi)_1,\ldots,L(\Pi)_n$ are mutually independent. 
\end{lemma}

Recall that the true underlying ranking of $\{\theta_1,\ldots,\theta_n\}$ is given by $\theta_{\Pi(1)} < \theta_{\Pi(2)} < \cdots < \theta_{\Pi(n)}.$ By the definition of Lehmer code, $L(\Pi^{-1})_i$ measures the number of items from $\{\theta_{i+1}\ldots,\theta_n\}$ that are smaller than $\theta_i$. Also notice that the $i$th insertion step of the variable-length noisy sorting code is using the auxiliary algorithm to estimate the number of items from $\{\theta_{n-i+1},\ldots,\theta_{n}\}$ that are smaller than $\theta_{n-i}$. Therefore, the $i$th insertion step of the code is essentially estimating the $(n-i)$th entry $L(\Pi^{-1})_{n-i}$ of the Lehmer code. Because there is a one-to-one mapping between permutations and Lehmer codes, every insertion step needs to be correct to output the correct estimation for $\Pi$. For $i\in[n-1]$, let $\mathcal{E}_i$ denote the event that estimation at $i$th insertion step is incorrect. We have
\begin{align*}
    \P(\Pi'=\Pi|\mathcal{A})&=\P(\mathcal{E}_1^c\cap\cdots\cap\mathcal{E}_{n-1}^c|\mathcal{A})\\
    &=\prod_{k=1}^{n-1}\P(\mathcal{E}_k^c \cond\mathcal{A}, \cap_{i=1}^{k-1}\mathcal{E}_i^c)\\
    &=\prod_{k=1}^{n-1}\left(1-\P(\mathcal{E}_k \cond\mathcal{A}, \cap_{i=1}^{k-1}\mathcal{E}_i^c)\right)\\
    &\ge 1-\sum_{k=1}^{n-1}\P(\mathcal{E}_k \cond\mathcal{A}, \cap_{i=1}^{k-1}\mathcal{E}_i^c). 
\end{align*}
By Lemma~\ref{lem:Lehmer-indep}, we know that each of $n$ entries of $L(\Pi^{-1})$ are uniformly distributed and they are independent of each other. Therefore, $\P(\mathcal{E}_i|\mathcal{A},\cap_{k=1}^{i-1}\mathcal{E}_k^c)$ represents the error probability of the variable-length BZ algorithm called at $i$th insertion step given a uniform prior for the correct position. By Lemma~\ref{thm:error-vlbz} and because that we set the tolerated error probability for each insertion step to $\frac{1}{n\log n}$, we know that $\P(\mathcal{E}_i|\mathcal{A}\cap_{k=1}^{i-1}\mathcal{E}_k^c)\le \frac{1}{n\log n}$ for each $i\in[n-1]$. Thus, we get that $\P(\Pi'\neq \Pi|\mathcal{A})\le\frac{1}{\log n}$, and it follows that
\[
\P(\hat{\Pi}=\Pi|\mathcal{A})\ge 1-\P(\Pi'=\Pi|\mathcal{A})-\P(\mathcal{B}^c|\mathcal{A})=1-o(1),
\]
which completes the proof.

\subsection{Proof of Lemma~\ref{lem:noisy_searching}}\label{app:pf-noisy-searching}
Let $u$ be the leaf node in the extended tree $T^{*}$ that represents the correct interval for $\td{\theta}$. Orient all the edges in $T^{*}$ towards the node $u$. By the definition of a tree, such an orientation is unique. Moreover, for each node $v \neq u$ in the tree, there exists exactly one edge connected to $v$ that is directed outwards from $v$, while all other edges connected to $v$ are directed towards $v$.

Recall that in Algorithm~\ref{alg:NST}, a repetition code of length $\beta$ is used for each comparison. At each node $v\neq u$ in $T^{*}$, at most three comparisons are made (two for verification of the interval and one for comparing with the median element). A sufficient condition for the random walk to go along the outgoing edge from $v$ is that all comparisons yield correct responses. Thus, the probability of this event is at least
\[
(1-\tau)^3
\]
where $\tau \triangleq \sum_{i=\ceil{\beta/2}}^{\beta}\binom{\beta}{i}p^i(1-p)^{\beta-i}$. Thus, by taking $\beta = \beta(p,\delta)$, we are guaranteed that this probability is at least $1-\delta$.

Now, to bound the probability of error of Algorithm~\ref{alg:NST}, let $S_\mathrm{f}$ denote the number of forward steps that go along the direction of an edge, and let $S_\mathrm{b}$ denote the number of backward steps that go against the direction of an edge. A sufficient condition for the algorithm to output the correct interval is that $S_\mathrm{f} - S_\mathrm{b} \geq \ceil{\log n}$, since the tree $T$ has depth $\ceil{\log n}$. This condition can be equivalently written as $S_\mathrm{f} - S_\mathrm{b} \geq \log n$ since $S_\mathrm{f}$ and $S_\mathrm{b}$ are integers. Since $S_\mathrm{f} + S_\mathrm{b} = s = \ceil{\frac{\tilde{m}}{3\beta(p,\delta)}}$, it follows that the probability of error of Algorithm~\ref{alg:NST} can be upper bounded by
\begin{align*}
    \P\left\{S_\mathrm{f} - S_\mathrm{b} < \log n\right\} &\leq \P\left\{S_\mathrm{b} > \tfrac{s-\log n}{2}\right\} \\
    &\leq \P\left\{\mathrm{Binom}\left(s,\delta\right)> \tfrac{s-\log n}{2}\right\}\\
    &\leq \exp\left(-sD\left(\frac12-\frac{\log n}{2s}\bigg|\bigg|\delta\right)\right),
\end{align*}
where the last inequality follows by condition~\eqref{eq:lem1_cond1} and the 
Chernoff bound~\cite[Theorem 1]{Hoeffding1963}.
\end{document}